\documentclass[10pt,twocolomun]{IEEEtran}

\usepackage{amsmath, graphics,amssymb,epsfig,subfigure,color,amsthm}

\usepackage{algorithm}
\usepackage{algorithmic}
\usepackage{float}
\usepackage{multirow}
\usepackage{cuted,flushend}
\usepackage{midfloat}
\newtheorem{theorem}{Theorem}
\newtheorem{proposition}{Proposition}
\newtheorem{lemma}{Lemma}

\newtheorem{corollary}{Corollary}

\usepackage{cite}

\def\b{\mathbb}

\def\SIR{\textrm{SIR}}

\def\b{\mathbb}

\def\d{{\rm d}}

\def\Nt{N_{\rm t}}

\def\ccdfCBdcond{{F}^{\rm c}_{\textrm{SIR} | \delta_1}(K,\Nt,\beta,\delta_1;\gamma)}

\DeclareMathOperator{\arccot}{arccot}

\title{Spectral Efficiency of Dynamic Coordinated Beamforming: A Stochastic Geometry Approach}
 
\author{\normalsize Namyoon Lee,
        David Morales-Jimenez,
        Angel Lozano,  
        and Robert W. Heath Jr.% <-this % stops a space
\thanks{N. Lee and R. W. Heath Jr. are with the Wireless Networking and Communications Group, Department of Electrical and Computer Engineering, The University of Texas at
Austin, Austin, TX 78712, USA. (e-mail:\{namyoon.lee, rheath\}@utexas.edu). Their work was supported in part by Huawei Technologies.}    
\thanks{D. Morales-Jimenez and A. Lozano are with the Department of Information and Communication Technologies, Universitat Pompeu Fabra, 08018 Barcelona, Spain. (e-mail:\{d.morales,angel.lozano\}@upf.edu). Their work was supported by the European Project FET 265578 "HIATUS" and by the MICINN Project TEC2012-34642.} 
\thanks{Parts of this paper are to be presented at the International Conference in Communications (ICC'14) \cite{ICC2014}}} % <-this % stops a space

\begin{document}

\maketitle

\begin{abstract} 
This paper characterizes the performance of coordinated beamforming with dynamic clustering. A downlink model based on stochastic geometry is put forth to analyze the performance of such base station (BS) coordination strategy. Analytical expressions for the complementary cumulative distribution function (CCDF) of the instantaneous signal-to-interference ratio (SIR) are derived in terms of relevant system parameters, chiefly the number of BSs forming the coordination clusters, the number of antennas per BS, and the pathloss exponent. Utilizing this CCDF, with pilot overheads further incorporated into the analysis, we formulate the optimization of the BS coordination clusters for a given fading coherence. Our results indicate that (i) coordinated beamforming is most beneficial to users that are in the outer part of their cells yet in the inner part of their coordination cluster, and that (ii) the optimal cluster cardinality for the typical user is small and it scales with the fading coherence. Simulation results verify the exactness of the SIR distributions derived for stochastic geometries, which are further compared with the corresponding distributions for deterministic grid networks.

\end{abstract}

%%%%%%%%%%%%%%%%%%%%%%%%%%%%%%%%%%%%%%%%%%%%%%%%%%%%%%%%%%%%%%%%%%%%%%%%%%%%%%%%%%%%%%%%%%%%%%%%%%
%------------------------------------------------------
%------------------------------------------------------
\section{Introduction}

%------------------------------------------------------
\subsection{Background}

Base station (BS) coordination is regarded as an effective approach to mitigate intercell interference \cite{CoMP:12,Bruno,Gesbert}. The idea is to allow multiple BSs to coordinate their transmit/receive strategies (e.g., beamforming, power control, and scheduling) by utilizing channel state information (CSI).
The performance would increase monotonically with the number of coordinated BSs if such CSI could be acquired at no cost and thus, ideally, entire systems should be coordinated  \cite{Foschini, Venkatesan, Cadambe_Jafar, Somekh, Huh}.
In practice though, coordination of an entire (large) network is not only computationally unfeasible, but undesirable once the ensuing overheads are taken into account \cite{Lozano_andrew_heath}. A central concept in the implementation of BS coordination is then that of a \emph{cluster}, defined as the set of BSs that a given user coordinates with. From the vantage of a user then, only those BSs outside the cluster are sources of interference. Intuitively, a larger cluster reduces intercell interference but it also increases the overheads required to acquire the necessary CSI at the BSs. It follows that determining the optimal cluster cardinality is a key step to assess the true benefits of BS coordination. This paper tackles such optimization for a particular coordination strategy.

%by, first, characterizing the signal-to-interference ratio (SIR) distribution in terms of relevant system parameters, and second, deriving the optimal number of cooperating BSs.

%------------------------------------------------------
\subsection{Related Work}

%There has been extensive work dealing with the capacity of cellular networks under various topology models. Approaches based on stochastic geometry are rapidly gaining momentum because of their superior analytical tractability and because they also match well the heterogeneous nature of emerging networks. In stochastic geometry, the location of the BSs is modeled as a Poisson point process (PPP) \cite{Andrews}.
%Within this framework, the performance of BS coordination schemes with fixed cluster structures established a-priori has been studied \cite{Huang_andrews,Baccelli,Salam}. Given the evidence (e.g., \cite{Gesbert_Dynamic,Mungara,Lee_Heath_Morales_Lozano_1}) that dynamic clusters based on user locations and channel propagation features perform far better than their fixed counterparts, there is clear interest in extending the existing stochastic geometry analyses to such dynamic cooperation structures.

%There has been extensive work dealing with the capacity of cellular networks under various topology models.
In toy setups where all the BSs can participate in the coordination, centralized schemes have been shown to yield sum spectral efficiencies that increase unboundedly with the transmit powers \cite{Foschini,Venkatesan,Cadambe_Jafar,Somekh}.
However, in large networks where the CSI-acquisition overheads and the channel uncertainty caused by fading selectivity prevent large cluster cardinalities, out-of-cluster interference is inevitable and the spectral efficiency has been shown to be fundamentally bounded no matter how sophisticated the cooperation \cite{Huh,Lozano_andrew_heath,Huang_TWC,Zhang_Andrews_Heath,Xu_Zhang_Andrews}. Nevertheless, small cooperation clusters not incurring too much overhead do provide performance improvements with respect to a noncooperative baseline.

Since, despite their regularity, deterministic grid models are remarkably unfriendly to analysis, most of the results on BS coordination for grid networks are simulation-based. Analytical results are available only for the simplest embodiments thereof, in particular for the so-called Wyner model where BSs and mobile users are located along a one-dimensional universe \cite{Xu_Zhang_Andrews, Simeone_book}.

Approaches based on stochastic geometry are rapidly gaining momentum because of their superior analytical tractability and because they happen to match well the heterogeneous nature of emerging networks \cite{Andrews}.
%In stochastic geometry, the location of BSs and users are modeled via a Poisson point process (PPP) \cite{Andrews}.
Within this framework, the performance of BS coordination schemes with fixed cluster structures established a-priori has been studied \cite{Huang_andrews,Baccelli,Salam}.

Dynamic BS clustering is a way of forming coordinated BS sets based on users' locations and channel quality. Given the evidence (e.g., \cite{Gesbert_Dynamic,Mungara,Lee_Heath_Morales_Lozano_1}) that dynamic clusters based on user locations and channel propagation features perform far better than their fixed counterparts, there is clear interest in extending the existing stochastic geometry analyses to such dynamic cooperation structures.

% \cite{Boccardi_Huang}
%------------------------------------------------------
\subsection{Contribution}

\begin{figure*}
\centering
\includegraphics[width=6.3in]{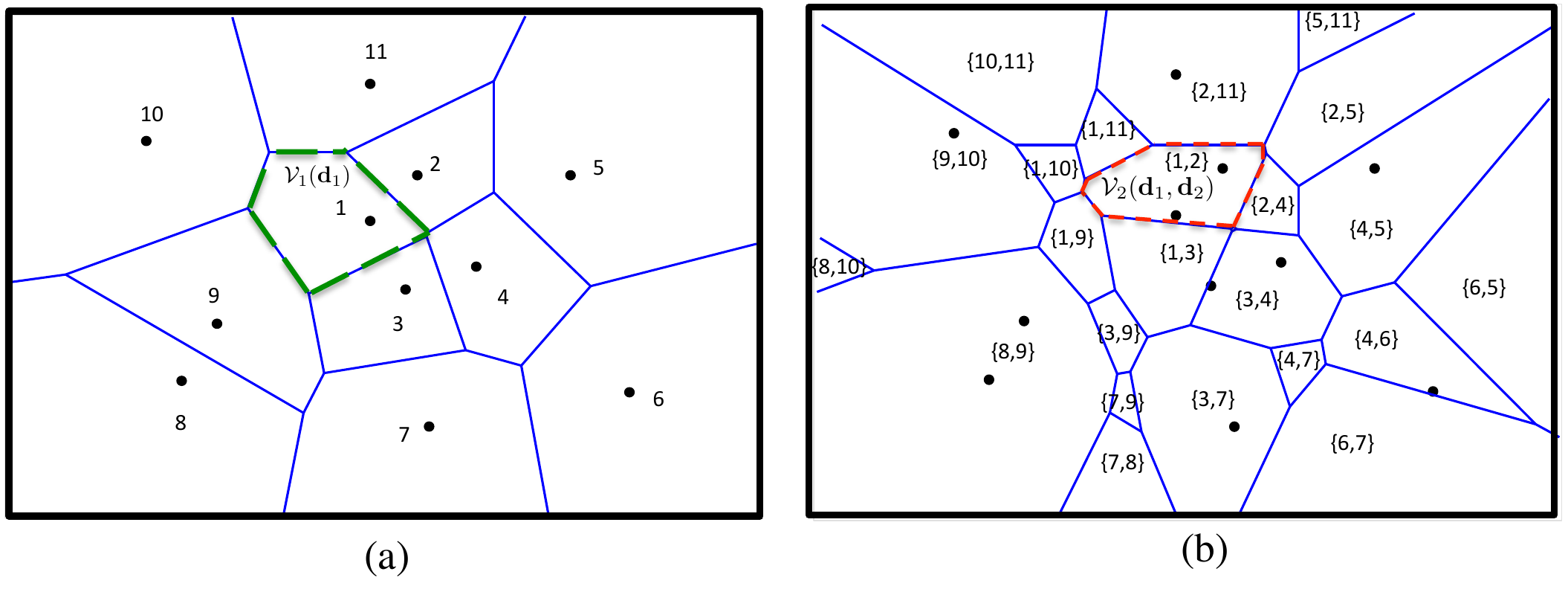}
\caption{Fig. 1(a) shows an instance of the first-order Voronoi tessellation on a two-dimensional plane. Each BS location ${\bf d}_i$ for $i\in\{1,2,\ldots,11\}$ is the center of a cell $\mathcal{V}_1({\bf d}_i)$. Fig. 1(b) illustrates the second-order Voronoi tessellation associated with the same BS locations in the left figure. Users in the region $\mathcal{V}_2 ({\bf d}_1, {\bf d}_2)$ choose to connect with BS 1 and BS 2.  } \label{fig:1}
\end{figure*}

We consider the downlink of a network whose topology is modeled through stochastic geometry. BS locations are modeled as a homogeneous Poisson point process (PPP) with the corresponding cells determined from a tessellation of the plane into Voroni regions. Users in each cell are randomly located and each one then defines its own cluster, i.e., the set of BSs it coordinates with, on the basis of such location. Under this dynamic BS clustering policy, the set of users that share the same BS cluster are served through \emph{coordinated beamforming}, a method that seeks to have each user communicate with one of the BSs in its cluster with no interference from all other in-cluster BSs.
%that seeks to elliminate intracluster interference. 
Our contribution is to characterize the performance of such dynamic coordinated beamforming.

We derive analytical expressions for the signal-to-interference ratio (SIR) distributions and the ergodic spectral efficiency in terms of system parameters, chiefly the pathloss exponent, the number of antennas per BS, the cluster cardinality, and the pilot overhead.
We obtain these analytical results for users with specific in-cluster relative locations and, by marginalizing over such locations, for the typical user.
%We first derive, in integral form, the complementary cumulative distribution function (CCDF) of the instantaneous SIR under the premise of a relative user location within its cluster.
 %of a typical user satisfying a particular in-cluster topology. An exact complementary cumulative distribution function (CCDF) of the instantaneous SIR is derived in integral form.
%To gain analytical tractability, closed-form upper and lower bounds of the SIR distribution are then derived in terms of relevant system parameters. With proper particularization of such parameters, the benefits of cooperation are gauged for different scenarios of interest. % such as, for example, when the user is at the cluster-center and cluster-edge.
%By leveraging the statistical distribution of the in-cluster topology, induced by the PPP, we also derive the CCDF of the SIR for a typical user, which provides a complete view about   spatially averaged SIR distribution.
Utilizing this latter result, we then characterize the benefits of coordination in terms of the net ergodic spectral efficiency, incorporating the pilot overheads required for coordinated beamforming. From this, we obtain the optimal cluster cardinality for the typical user. Our finding is that coordinated beamforming is most beneficial to users that are in the outer  part of their cells yet in the inner part of their coordination cluster, and that the optimal cluster cardinality for the typical user is small and scales with the fading coherence
Through simulation, the accuracy of the derived SIR distributions is verified.

The remainder of the paper is organized as follows. Section II describes the proposed models as well as the performance metrics for the considered coordinated beamforming scheme. In Section III, analytical expressions for the SIR distribution are derived for specific relative in-cluster locations. The SIR distribution for the typical user is derived in Section IV, and then utilized in Section V to analyze the optimal cluster cardinality. Section VI concludes the paper.

%------------------------------------------------------
%------------------------------------------------------
\section{Models and Metrics}

%------------------------------------------------------
\subsection{Network Model}

%We consider a cellular network consisting of a large number of BSs, each equipped with $N_{\textrm{t}}$ antennas. The locations of these BSs are established according to a homogeneous PPP, $\Phi=\{x_i,i\in \mathbb{N}\}$, on the plane $\mathbb{R}^2$ with density $\lambda$. By tessellating the plane into Voronoi regions around each BS, we define cells in the traditional manner.
%and each point constructs a (Voronoi) cell as a random spatial tessellation.
%Further, the population of single-antenna mobile users is distributed according to another homogeneous PPP, $\Phi_{\textrm{U}}=\{y_i,i\in \mathbb{N}\}$, which has density
%$\lambda_{\textrm{U}}$ and is independent of $\Phi$.

We consider a cellular network consisting of BSs, each equipped with $N_{\textrm{t}}$ antennas. The locations of these BS are established according to a homogeneous PPP, $\Phi=\{{\bf d}_k,k\in \mathbb{N}\}$, on the plane $\mathbb{R}^2$. This PPP has density $\lambda$. By tessellating the plane into Voronoi regions around each BS, we can define cells in the traditional manner.
%and each point constructs a (Voronoi) cell as a random spatial tessellation.
We further consider a population of single-antenna users distributed according to another homogeneous PPP, $\Phi_{\textrm{U}}=\{{\bf u}_k,k\in \mathbb{N}\}$, which has density
$\lambda_{\textrm{U}}$ and is independent of $\Phi$.

%------------------------------------------------------

\subsection{Dynamic Clustering Model}
 %Without loss of generality, we can order the BSs in increasing distance from the user. 

The dynamic clustering technique that we analyze relies on the policy of having each user coordinate with the $K$ closest BSs, where $K\leq N_{\rm t}$ is the cardinality of the coordination clusters. With this clustering policy, the set of users' locations that can be served by their $K$ nearest BSs is formally defined by the notion of $K$th-order Voronoi tessellation. Denoted by ${\mathcal{V}_K}({\bf d}_{1}, \ldots,{\bf d}_{K})$, the $K$th-order Voronoi cell associated with $K$ distinct points ${\bf d}_{1}, \ldots,{\bf d}_{K}$ is the set of points closer to ${\bf d}_{1}, \ldots,{\bf d}_{K}$ than to any other point of $\Phi$, i.e.,
\begin{align}
&{\mathcal{V}_K}({\bf d}_{1}, \ldots,{\bf d}_{K})\nonumber\\&=\left\{{\bf d}\in\mathbb{R}^2 \mid \cap_{k=1}^K\left\{\|{\bf d}-{\bf d}_{{k}}\|_2 \leq \|{\bf d}-{\bf d}_{j}\|_2 \right\} \right\}
\end{align}
where $\forall {\bf d}_j\in \Phi/\{{\bf d}_{1}, \ldots,{\bf d}_{K}\}$.
Taking $K=2$ as an example, Fig. \ref{fig:1} depicts the second-order Voroni tessellation with the corresponding Voroni cells for the proposed dynamic clustering method. Users in the region $\mathcal{V}_2 ({\bf d}_1, {\bf d}_2)$ connect to the cluster formed by BSs at ${\bf d}_1$ and ${\bf d}_2$. Further, the area of the $K$th-order Voronoi cell is nonzero with probability one. Then, provided that the user density is much higher than the BSs density, i.e., $\lambda_U \gg \lambda$, there will be (with high probability) at least $K$ users in ${\mathcal{V}_K}({\bf d}_{1}, \ldots,{\bf d}_{K})$ choosing to connect with the $K$ BSs at ${\bf d}_{1}, \ldots,{\bf d}_{K}$. Each BS serves one user per time-frequency resource, for a total of $K$ users per cluster, with the remaining users in the network accommodated in different signaling resources.

%We further define a \emph{clustering pattern} $\mathcal{P}$ as any set of nonoverlapping $K$th-order Voroni cells, i.e.,
%\begin{align}
%\mathcal{P}=\{ \mathcal{V}_K^{(i)}( \mathcal{D}_i ) \}_{i=1}^{\infty}
%\end{align}
%with $\mathcal{D}_i$ being a nonoverlapping partition of $\Phi$ in subsets of size $K$. Note that one clustering pattern is not enough to accommodate all users across the network and, therefore, different patterns can be defined for different time-frequency resources. Taking $K=2$ as an example, Fig. \ref{fig:1} depicts the second-order Voroni tessellation with the corresponding Voroni cells for the proposed dynamic clustering method. One possible clustering pattern is
%\begin{align}
%\mathcal{P}_1=\{ \mathcal{V}_2({\bf d}_{1},{\bf d}_{2}), \mathcal{V}_2({\bf d}_{3},{\bf d}_{4}),\mathcal{V}_2({\bf d}_{5},{\bf d}_{6}),\mathcal{V}_2({\bf d}_{7},{\bf d}_{8})\mathcal{V}_2({\bf d}_{9},{\bf d}_{10})\} ,
%\end{align}
%which accommodates any two users in $\mathcal{V}_2({\bf d}_{k},{\bf d}_{k+1})$, $k\in\{1,3,5,7,9\}$, whilst other users outside $\mathcal{P}_1$ can be served through a different pattern defined in a different time-frequency resource.

For later use, we also introduce a geometric parameter $\delta_1=\frac{\|{\bf d}_1\|_2}{\|{\bf d}_K\|_2}$ defined as the distance to the closest BS normalized by the distance to the furthest BS in the cluster. This parameter plays an important role in interpreting the SIR distribution and the ergodic spectral efficiency for any specific in-cluster geometry, i.e, for any specific user locations within the cluster. Specifically, a smaller $\delta_1$ implies a larger \emph{protection area} and vice versa, with this area being the minimum that is sure to be free of out-of-cluster interfering BSs. For example, when two BSs at ${\bf d}_1$ and ${\bf d}_2$ serve two users in $\mathcal{V}_2({\bf d}_1, {\bf d}_2)$ through coordinated beamforming, the user associated with the BS in ${\bf d}_1$ has a protection area $A=\pi(\|{\bf d}_2\|_2^2-\|{\bf d}_1\|_2^2)=\frac{\pi}{\|{\bf d}_2\|_2^2}  (1-\delta_1^2)$. %This area corresponds to the area which is free of interference sources, i.e., a smaller $\delta_1$ (larger $A$) means further sources of out-of-cluster interference. Since no BSs in the protection area $A$ cause interference to the user associated with the BS1, the smaller (larger) $\delta_1$ implies the larger (smaller) interference protection region.

\subsection{Signal Model}

Under the premise of separate encoding at each BS, the $k$th BS sends an information symbol $s_k$ (intended for the $k$th user) through a linear beamforming vector ${\bf v}_k=[v_{k}^1,v_{k}^2,\ldots,v_{k}^{N_{\rm t}}]^T$ with unit norm, $\|{\bf v}_k\|_2=1$, $k\in\{1,2,\ldots,K\}$. Without loss of generality, let us focus on a user located at the origin. The observation at this user is
\begin{align}
{y}_1 &=  {\|{\bf d}_1\|}^{-\beta/2}{\bf h}_{1,1} {\bf v}_{1}s_1 + \sum_{k = 2}^{K} {\|{\bf d}_k\|}^{-\beta/2}{\bf h}_{1,k} {\bf v}_{k}s_k \nonumber\\ & + \sum_{k=K+1}^{\infty} {\|{\bf d}_k\|}^{-\beta/2}{\bf h}_{1,k} {\bf v}_{k}s_k + z_1
\end{align}
where ${\bf h}_{1,k}=[h_{1,k}^1,h_{1,k}^2,\ldots,h_{1,k}^{N_{\rm t}}]\in\mathbb{C}^{1\times N_{\rm t}}$ represents the downlink channel between the $k$th BS and the user, with entries that are independent and identically distributed (IID) complex Gaussian random variables having zero mean and unit variance,
i.e., $\mathcal{CN}(0,1)$. The channels vary over time in an IID block-faded fashion. Further, $\beta$ represents the pathloss exponent and ${z_1}$ denotes the additive Gaussian noise, $z_1 \sim \mathcal{CN}(0,\sigma^2)$. The transmit power at each BS satisfies $\mathbb{E}\left[|{s}_{k}|^2\right] \leq P$.

Each user learns the downlink channels from the $K$ BSs within its cluster by means of orthogonal pilot symbols and then conveys this information back to the BSs via error-free feedback links.
%\subsection{Signal-to-Interference (SIR) Ratio}
%Assuming each BS has acquired perfect CSI for all the users within the cluster, the main idea of
From this CSI, the coordinated beamforming scheme constructs the beamforming vectors ${\bf v}_k$, $k \in \{ 1,\ldots,K \}$, that nullify intra-cluster interference while maximizing the desired signal strength for the $K$ users in the cluster. Thus, the $k$-th BS selects beamforming vector ${\bf v}_k$ solving
\begin{align}
\textrm{maximize:} \quad& |{\bf h}_{k,k} {\bf v}_k|^2 \label{eq:BFsol}  \\
\textrm{subject to:} \quad& {\bf h}_{i,k} {\bf v}_k =0 \quad \textrm{for} ~~i\neq k, 
 \nonumber \\
& \|{\bf v}_k\|_2=1 \nonumber
\end{align}
which always exists when $N_{\textrm{t}} \geq K$. The corresponding instantaneous SIR for the user 1 at the origin is
\begin{align}
\textrm{SIR}(K,N_{\rm t},\beta)=\frac{|{\bf h}_{1,1} {\bf v}_1|^2  {\|{\bf d}_1\|}^{-\beta}} {I_K} .
\label{eq:SIRCB}
 \end{align}
 where
\begin{equation}
I_K = \sum_{k=K+1}^{\infty}|{\bf h}_{1,k} {\bf v}_{k}|^2{\|{\bf d}_k\|}^{-\beta}
\end{equation}
is the aggregate out-of-cluster interference power. This instantaneous SIR in (\ref{eq:SIRCB}) involves multiple levels of randomness:
\begin{enumerate}
\item The randomness associated with the user location relative to its serving BS; this is incorporated through $\|{\bf d}_1\|$. Equivalently, and more conveniently to the analysis that follows later, it can be incorporated through $\|{\bf d}_K\|$ and $\delta_1=\frac{\|{\bf d}_1\|}{\|{\bf d}_K\|}$.
\item The randomness associated with the user location relative to the interfering BSs; this is incorporated through $\|{\bf d}_k\|$, $k > K$.
\item The randomness associated with the desired link fading; this is incorporated through ${\bf h}_{1,1}$.
\item The randomness associated with the interference fading; this is incorporated through ${\bf h}_{1,k}$, $k > K$.
\end{enumerate}

\subsection{Performance Metrics}

The CCDF of the instantaneous SIR is characterized at two different levels, with the absolute dimensions of the network abstracted out.
% and in each case an interpretation of the significance is provided.
 % is a performance indicator on its own and a stepping stone towards the computation of the spectral efficiency.
%Commonly, the distribution of the SIR in (\ref{eq:SIRCB}) would be obtained after averaging over space, i.e., over all possible BSs locations, so that it would capture the typical user (ergodic) performance. Here, the cooperation scheme induces a new spatial structure, the cluster, and the performance of such scheme turns out to be very sensitive to the relative location of the user within such cluster.
%Therefore, the SIR distribution is considered at two levels of spatial averaging depending on the knowledge of the in-cluster topology (relative location of the user within the cluster).
 
\subsubsection{Specific Relative Cluster Geometry} 

First, we characterize the CCDF of the SIR for some given relative distances $\{\delta_1, \ldots, \delta_K\}$, but with the absolute distances $\{{\bf d}_1,\ldots,{\bf d}_k\}$ and the out-of-cluster interference $I_K$ marginalized over.
Since the signals received from BSs $k=2,\ldots,K$ do not contribute interference by virtue of (\ref{eq:BFsol}), it suffices to condition on $\delta_1$ and the ensuing conditional CCDF is
\begin{align}
&\ccdfCBdcond =\mathbb{P}\left[\textrm{SIR}(K,N_{\rm t},\beta)\geq \gamma \mid \delta_1 \right]  \\
&= \mathbb{E}\left[\mathbb{P}\left[\frac{|{\bf h}_{1,1} {\bf v}_1|^2\left({\delta_1}{\|{\bf d}_K\|}\right)^{-\beta}}{I_K}\geq \gamma  \mid \delta_1 , \|{\bf d}_K\|, I_K\right]  \!\!\mid \delta_1 \!\right]\!\!.\label{eq:metric_SIR_distri}
\end{align}
where the expectation over $\|{\bf d}_K\|$ and $I_K$, characterized in Section \ref{sec:deterministicTopology}, effect the marginalization.
%This CCDF can be interpreted as the coverage probability with target SIR $\gamma$ conditioned on $\delta_1$, $0\leq \delta_1\leq 1$, the relative distance to the closest serving BS.
This conditional CCDF does not correspond to the distribution of the SIR experience by any actual user in the system, but it is representative of the average behavior in all possible cluster geometries that share a particular $\delta_1$.

\subsubsection{Average Cluster Geometry} By further marginalizing over $\delta_1$, we obtain the SIR distribution averaged over all possible geometries, which is less informative than the one in (\ref{eq:metric_SIR_distri}). In particular, this fully marginalized distribution does not allow discriminating between situations that are either favorable or adverse to coordinated beamforming, but it does serve as a stepping stone towards the computation of average quantities.
The fully marginalized CCDF of the SIR is
\begin{align}
&F^{\rm c}_{\textrm{SIR}} (K,\Nt,\beta;\gamma)=\mathbb{P}\left[\textrm{SIR}(K,N_{\rm t},\beta)\geq \gamma \right]      \\
&=\mathbb{E}\left[\mathbb{P}\left[\frac{|{\bf h}_{1,1} {\bf v}_1|^2\left({\delta_1}{\|{\bf d}_K\|}\right)^{-\beta}}{I_K}\geq \gamma\mid  \delta_1, \|{\bf d}_K\|, I_K \right]  \right]\!\!.
\label{eq:metric_avg_SIR_distri} 
\end{align}
where the expectation is now also over $\delta_1$ in addition to $\|{\bf d}_K\|$ and $I_K$. This distribution will be used to characterize the performance of the typical user in Secs. \ref{sec:randomTopology} and \ref{sec:optimalSize}.

\section{Specific Relative Cluster Geometry}
\label{sec:deterministicTopology}

In this section, we characterize the conditional CCDF in (\ref{eq:metric_SIR_distri}) %in integral form and, for further tractability, we also provide tight closed-form lower and upper bounds
in terms of $K$, $\Nt$ and $\beta$. 

%------------------------------------------------------
\subsection{General Characterization}

%The exact characterization relies on two Lemmas, reproduced next for the sake of completeness.
%Lemma \ref{lemma:fading} informs about the distribution of the effective channel fading for the both the desired and the interference links when the coordinated beamforming strategy in (\ref{eq:BFsol}) is employed. Lemma \ref{lemma:distance}, in turn, provides the distribution of the distance to the $K$th BS when the BS locations are determined by a 2-D homogeneous PPP.

%Leveraging Lemmas \ref{lemma:fading} and \ref{lemma:distance}, an exact expression for the  conditional CCDF is given next.

We begin by providing a general characterization in integral form.

\begin{theorem} \label{Th1}
For a given $\delta_1$,
\begin{align}
F^{\rm c}_{{\rm SIR}|\delta_1}(K,N_{\rm t},\beta,\delta_1;\gamma)\!=\!  \mathbb{E}\!\left[\!\sum_{m=0}^{\Nt-K}\!\! \frac{r^{\beta m}}{m !}(\!-\!1)^{m}\!\!\! \left. \frac{\d^m\mathcal{L}_{{\tilde I}_r}(s)}{\d s^m} \right|_{s=r^\beta} \!\right]
\end{align}
where ${\tilde I}_{r}=\delta_1^{\beta}\gamma \sum_{k=K+1}^{\infty} H_k\|{\bf d}_k\|^{-\beta}$ while $\mathcal{L}_{{\tilde I}_r}(s)=\mathbb{E}\left[e^{-s{\tilde I}_r}\right]$ denotes the Laplace transform of ${\tilde I}_r$, which is given in (\ref{eq:conditional_LF}) in Appendix B, and the expectation is over $\|{\bf d}_K\|=r$, distributed as per Lemma \ref{lemma:distance}.
\end{theorem}

\begin{proof}
See Appendix \ref{proof:Th1}.
\end{proof}

%It involves the derivatives of the Laplace transform of $\tilde{I}_r$ and the expectation over $\|x_K\|$, therefore making any further analysis untractable. In what follows, tight lower and upper bounds for the SIR distribution are derived in closed form, thus allowing for a tractable analysis of the multi-cell coordination gain in terms of relevant system parameters.

%------------------------------------------------------
%\subsection{Upper and Lower Bounds}

%The following theorem provides closed-form expressions of the upper and lower bounds for the CCDF of the SIR. 
Although general and exact, the form given in Theorem \ref{Th1} is rather unwieldy, motivating the interest in more compact characterizations.
Still in full generality, we next provide closed-form upper and lower bounds to the distribution.

\begin{theorem} \label{Th2}
For a given $\delta_1$,
\begin{align}
F^{\rm c,L}_{{\rm SIR}|\delta_1}(K,N_{\rm t},\beta,\delta_1;\gamma) \leq F^{\rm c}_{{\rm SIR}|\delta_1}(K,N_{\rm t},\beta,\delta_1;\gamma)\nonumber\\
F^{\rm c}_{{\rm SIR}|\delta_1}(K,N_{\rm t},\beta,\delta_1;\gamma) \leq F^{\rm c,U}_{{\rm SIR}|\delta_1}(K,N_{\rm t},\beta,\delta_1;\gamma)
\end{align}
with
\begin{align}
F^{\rm c,U}_{{\rm SIR}|\delta_1}(K,N_{\rm t},\beta,\delta_1;\gamma)
 &\!=\!\!\!\!\! \sum_{\ell=1}^{N_{\rm{t}}-K+1}\frac{\binom{N_{\rm{t}}-K+1}{\ell}(-1)^{\ell+1}}{\left[1+\mathcal{D}(\ell \kappa \delta_1^{\beta} \gamma, \beta) \right]^{K}} \label{eq:Cov_upper}  \\
F^{\rm c,L}_{{\rm SIR}|\delta_1}(K,N_{\rm t},\beta,\delta_1;\gamma)
 &\!=\!\!\!\! \sum_{\ell=1}^{N_{\rm{t}}-K+1}\frac{\binom{N_{\rm{t}}-K+1}{\ell}(-1)^{\ell+1}}{\left[1+\mathcal{D}(\ell\delta_1^{\beta}  \gamma, \beta)\right]^{K}} \label{eq:Cov_lower}
\end{align}
where $\kappa=(N_{\rm{t}}-K+1)! ^{\frac{-1}{N_{\rm{t}}-K+1}}$ and
\begin{equation}
\mathcal{D}(A, B)=\frac{2A}{B-2} \,  {}_2 F_1 \! \left(1,1-\frac{2}{B},2-\frac{2}{B},-A\right)
\end{equation}
with ${}_2F_1(\cdot)$ the Gauss hypergeometric function.

%For $\Nt=K$ the upper and lower bounds coincide, thus giving the exact CCDF
%\begin{align}
%{F}^{\rm c}_{\textrm{SIR}}(K,K,\beta,\delta_1;\gamma) &= \left[\frac{1}{1+\mathcal{D}( \delta_1^{\beta} \gamma, \beta)} \right]^{K}. \label{eq:Cov_M1} 
%\end{align}

\end{theorem}

\begin{proof}
See Appendix \ref{proof:Th2}.
\end{proof}

The upper and lower bounds coincide when $N_{\rm t}=K$, implying that for this most important case we obtain the exact CCDF. 

\begin{corollary}
\label{Messi}
For a given $\delta_1$ with $K = \Nt$,
\begin{align}
F^{\rm c}_{\textrm{SIR}|\delta_1}(K,K,\beta,\delta_1;\gamma)
& =   \!\frac{1}{\left[1+\mathcal{D}(\delta_1^{\beta} \gamma, \beta)\right]^{K}    }.\label{eq:Cov_prob_BS_coop}
\end{align}
\end{corollary}
\subsection{Special Cases}

To shed further light on the significance of the expressions in Thm. \ref{Th2} and Cor. \ref{Messi}, it is instructive to consider certain special cases.

\subsubsection{Noncoordinated Network}

The most basic special case is the one where there is no coordinated beamforming, i.e., where $K=1$ (and $\delta_1=1$ with the conditioning thereupon immaterial).
By setting $\Nt=1$ we then recover the CCDF of the SIR given in \cite{Andrews}, namely
\begin{align}
F^{\rm c}_{\textrm{SIR}}(1,1,\beta;\gamma)=\frac{1}{1+\mathcal{D}(  \gamma, \beta)}
   \label{eq:CovProb_NoCoop_signle}
\end{align}
which Thm. \ref{Th2} therefore generalizes.
For this special case, the derived expressions are useful to characterize the benefits of having $N_{\rm t}$ antennas. As illustrated in Fig. \ref{fig:multipleantenna}, the upper bound  tightly matches the exact CCDF over the entire range of SIRs of interest and for distinct values of $N_{\rm t}$, i.e., $F^{\rm c,U}_{\textrm{SIR}} (1,\Nt,\beta,1;\gamma) \simeq F^{\rm c}_{\textrm{SIR}} (1,\Nt,\beta,1;\gamma) $.

% The lower bound $F^{\rm c,L}_{\textrm{SIR}} (1,\Nt,\beta,1;\gamma)$ corresponds to the exact SIR distribution when an antenna selection scheme that selects the best channel gain is employed. Thus, the exact CCDF of the SIR achieved by such an antenna selection scheme is
% \begin{align}
%F^{\rm c}_{\textrm{SIR}_{\rm AS}}(1,\Nt,\beta,1;\gamma)= \sum_{\ell=1}^{N_{\rm t}}\binom{N_{\rm t}}{\ell}(-1)^{\ell+1}\frac{1}{1+\mathcal{D}(\ell \gamma, \beta)}.
%   \label{eq:CovProb_NoCoop_antenna_selection}
%\end{align}
%The following example shows how this expression illuminates the way in which the SIR distribution varies with the number of transmit antennas. Specifically, we can identify the performance gain attained with an additional antenna as a function of the pathloss exponent $\beta$ and the target SIR $\gamma$ in a non-cooperative network.
%
%\textbf{Example 1}: When $N_{\rm t}=2$, the CCDF of the SIR (coverage probability) for the antenna selection scheme is
%\begin{align}
%F^{\rm c}_{\textrm{SIR}_{\rm AS}}(1,2,\beta,1;\gamma)&= \frac{2}{1+\mathcal{D}( \gamma, \beta)}-\frac{1}{1+\mathcal{D}(2 \gamma, \beta)} \nonumber \\
%&= F^{\rm c}_{\textrm{SIR}}(1,1,\beta,1;\gamma)+ \underbrace{\frac{1}{1+\mathcal{D}( \gamma, \beta)}-\frac{1}{1+\mathcal{D}(2 \gamma, \beta)}}_{\textrm{The gain term by the antenna selection}}.
%   \label{eq:Two_antenna_selection}
%\end{align}

\begin{figure}
\centering
\includegraphics[width=3.7in]{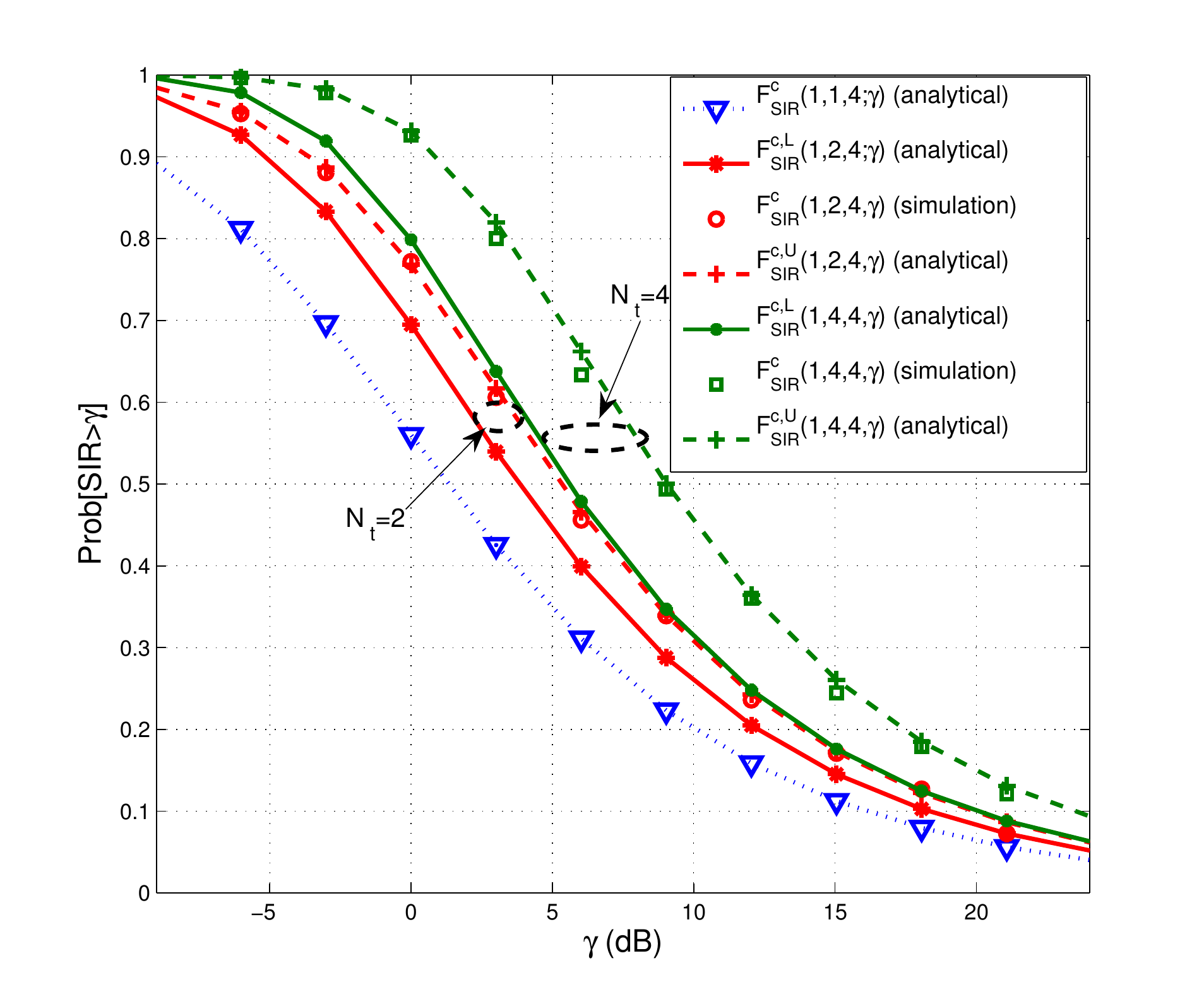}\vspace{-0.3cm}
\caption{CCDF of the SIR for $\beta=4$ and different numbers of transmit antennas in the absence of coordinated beamforming.} \label{fig:multipleantenna}\vspace{-0.3cm}
\end{figure}

\subsubsection{Coordinated Network with  $N_{\rm t}=K$}

With the coordinated beamforming activated and $N_{\rm t}=K$, the behavior is characterized by the simple expression in Cor. \ref{Messi}.
%\begin{align}
%F^{\rm c}_{\textrm{SIR}|\delta_1}(K,K,\beta,\delta_1;\gamma)
%& =   \!\frac{1}{\left[1+\mathcal{D}(\delta_1^{\beta} \gamma, \beta)\right]^{K}    }.\label{eq:Cov_prob_BS_coop}
%\end{align}
Particularized to $\beta=4$, which is a standard value in terrestrial outdoor wireless systems, $\mathcal{D}(\cdot, \cdot)$ reduces to
\begin{equation}
\mathcal{D}(\xi, 4) = \sqrt{\xi}  \arccot \left( \frac{1}{\sqrt{\xi}} \right)
\label{eq:beta4}
\end{equation}
which, plugged into (\ref{eq:Cov_prob_BS_coop}), yields
\begin{align}
F^{\rm c}_{\textrm{SIR}|\delta_1}\left(K,K,4,\delta_1;\gamma \right)
& =   \!\frac{1}{\left[1+ \sqrt{\gamma}\delta_1^2\arccot\left(\frac{1}{\sqrt{\gamma}\delta_1^2} \right)\right]^{K}    }.\label{eq:Cov_prob_BS_coop_simple}
\end{align} 
This simple CCDF %of the SIR can be interpreted as the coverage probability when the user's relative location within the cluster satisfies $\delta_1=\frac{\|{\bf d}_1\|}{\|{\bf d}_k\|}$, which
facilitates gauging different scenarios as indicated earlier: $\delta_1 \ll 1$ corresponds to users well isolated from out-of-cluster interference (i.e., with a large protection area) whilst $\delta_1 \approx 1$ corresponds to users susceptible to strong out-of-cluster interference (small protection area). Specifically, as illustrated in Fig. \ref{fig:SIR_dist_cond}, when $K=N_{\rm t}=2$ and $\beta=4$, a user with $\delta_1=0.5$ has a better SIR distribution---and thus higher benefits from coordinated beamforming---than a user with $\delta_1=0.8$. % because the user having $\delta_1=0.5$ is relatively located at the cluster-center than the user with $\delta_1=0.8$; thereby it obtains more benefits of BS coordination than the cluster-edge user.
A similar observation can be made for $\delta_1=1/4$ and $\delta_1=1/3$ when $K=N_{\rm t}=\beta=4$. Further, we compare the SIR distributions derived for the given specific in-cluster geometry $\delta_1\in\{0.5,0.8\}$ against those without coordination, i.e., for $K=1$ and  $N_{\rm t}=2$ (and hence for $\delta_1=1$). As illustrated in Fig. \ref{fig:SIR_dist_cond}, on the one hand the user with a small protection area ($\delta_1=0.8$) has a worse SIR distribution than that without BS coordination for all range of $\gamma$ and, on the other hand, the user with a larger protection area ($\delta=0.5$) exhibits better SIR distributional results up to $\gamma=16$ dB.
Agreeing with intuition, if a small protection area is created by coordinated beamforming, it might be better to use the two antennas for maximizing the desired signal power instead of canceling the nearest interferer. Conversely, in situations with a large protection area, canceling the nearest interference yields more benefits than maximizing the desired signal power.   

%This shows that a user with a small interference protection region (a large $\delta$) does not obtain the benefits of BS coordination for all range of $\gamma$.
%(this would agree with the intuition that, if the protection region is small, it might be better to use the two antennas for maximizing the desired signal power instead of canceling the nearest-interferer...because the second strongest interf. is still very close).

\begin{figure}
\centering
\includegraphics[width=3.7in]{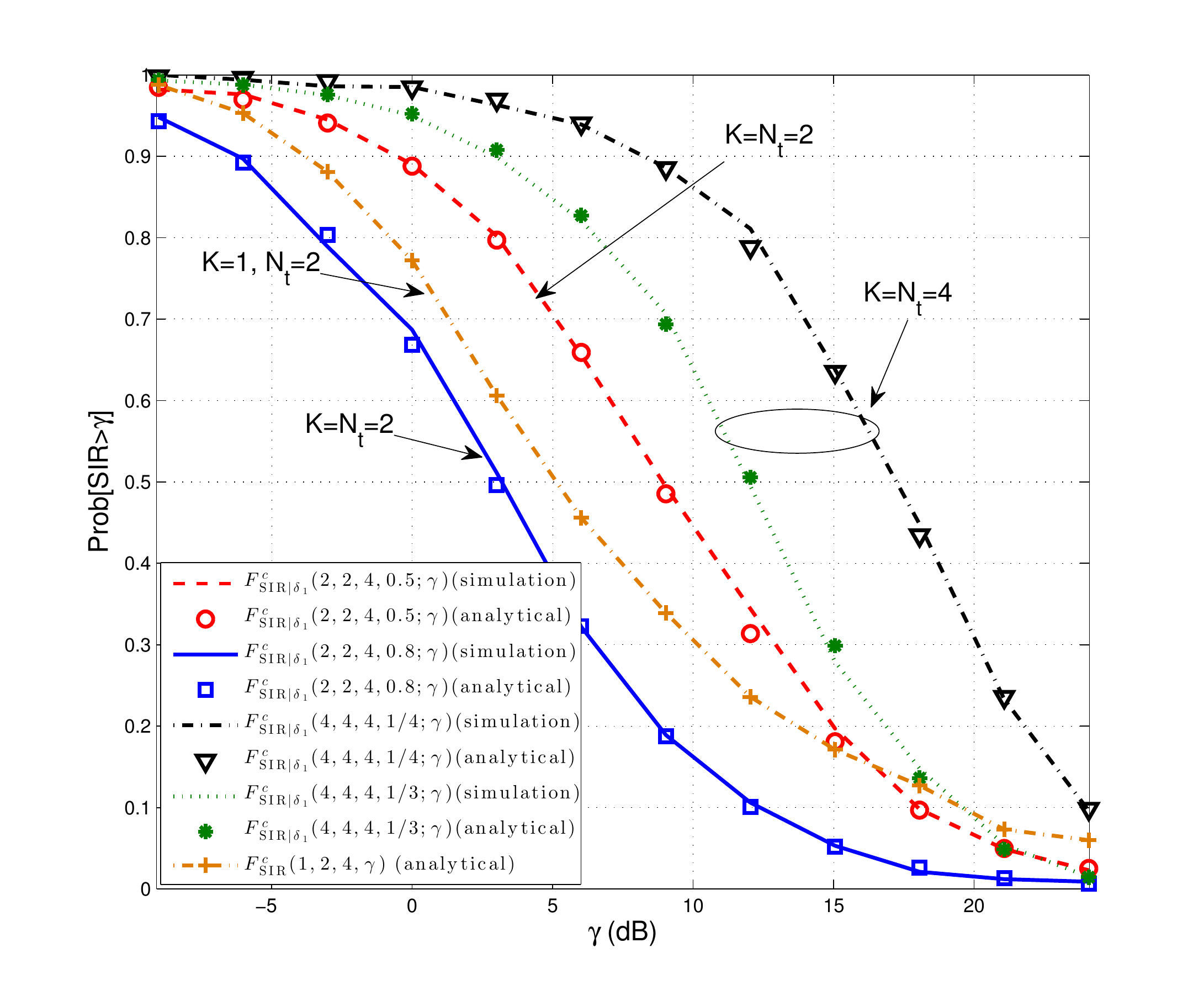}\vspace{-0.5cm}
\caption{Conditional CCDF of the SIR for $\beta=4$ and $K=N_{\rm t}\in\{2,4\}$.} \label{fig:SIR_dist_cond}\vspace{-0.3cm}
\end{figure}
%I DONT GET THE NEXT POINT. WITH $\delta=1$ AND $\delta \neq 1$ WE HAVE THE SAME AVERAGE OVER POSSIBLE LOCATIONS. THE ONLY DIFFERENCE IS THAT WITH   $\delta \neq 1$ (K>1) WE HAVE A RESTRICTION IN THE RELATIVE DISTANCES BETWEEN BSs...BUT IF K=1, THEN THERE IS NO RESTRICTION. PLEASE CHECK THIS POINT. IF YOU AGREE, THEN WE HAVE TO REMOVE THIS THING ABOUT THE BIASED RESULT (ALSO IN THE REVIEWER'S RESPONSES)
%Nevertheless, it is worthwhile to mention that the SIR distribution for the scenario of no BS coordination ($\delta_1 = 1$) in Fig. \ref{fig:SIR_dist_cond} is the averaged SIR distribution over all possible user locations in the cell. Therefore, the direct comparison with the SIR distribution condition on the in-cluster geometry $\delta_1$ for the BS coordination scenario may provide a biased result depending on the value of $\delta_1$. To verify the benefits of coordinated beamforming, we will compare the SIR distributions and ergodic spectral efficiency for an average cluster geometry with those attained by the no BS coordination case in a sequel.

The following example further illustrates the benefits of coordinated beamforming in terms of ergodic spectral efficiency for different relative locations, $\delta_1$.

{\em Example 1:} For $N_{\rm t}=K=2$ and $\beta=4$, the ergodic spectral efficiency (in bits/s/Hz) is 
 \begin{align}
C(2,2,4,\delta_1) &\! =\!\int_{0}^{\infty}\!\!\! \log_2 (1+\gamma) \, \d F_{\rm SIR | \delta_1}(2,2,4,\delta_1;\gamma) \\
\label{Iniesta}
&\! =\! \int_{0}^{\infty} \frac{\log_2 e}{(1+\gamma)} F^{\rm c}_{\rm SIR | \delta_1}(2,2,4,\delta_1;
\gamma) \d \gamma \\
&\!\!\!\!\!\!\!\!\!\!\!\!\!\!\!\!\!\!\!\! = \int_{0}^{\infty}  \!\frac{\log_2 e}{(1+\gamma)\left[1\!+\! \sqrt{\gamma}\delta_1^2\arccot\left(1/\sqrt{\gamma}\delta_1^2 \right)\right]^{2} } \d \gamma\label{eq:object_func}
\end{align}
where in (\ref{Iniesta}) integration by parts was applied.
The values for different relative locations $\delta_1$, listed in Table \ref{table1}, reveal that the most substantial gains of coordinated beamforming are obtained for small $\delta_1$. This is because % the cluster-center user obtains large benefits by BS clustering as it receives small interference power from out-of-cluster due to BS coordination. Meanwhile, the
 users with large $\delta_1$ remain subject to strong interference from out-of-cluster BSs.

\begin{table}
\caption{Ergodic Spectral Efficiency (bits/s/Hz) of Coordinated Beamforming for $N_{\rm t}=K=2$. }\vspace{-0.1cm}
\centerline{
    \begin{tabular}{c|c|c|c}
	\hline
	 Relative location $\delta_1$ & $\delta_1=\frac{1}{3}$ & $\delta_1=\frac{1}{2}$ & $\delta_1=\frac{2}{3}$   \\
	\hline \hline
	Ergodic spectral efficiency $C(2,2,4,\delta_1)$ & 5.377 & 3.3361 & 2.1318    \\ \hline 
    \end{tabular}}\label{table1}\vspace{-0.1cm}
\end{table}

 %Fo example, when $\delta_1=\frac{1}{2}$, the user located around the cluster-center area can be worse if the second nearest BS does not perform the coordination, i.e., $K=1$. When the second nearest BS performs the coordination, i.e., $K=2$, the user becomes to be located at the cluster-center area and obtains the large benefits by the BS clustering. Meanwhile, the user placed at the cluster-edge obtains relatively less benefits by the BS clustering because it receives still receives the interference signals from out-of-cluster.  

%------------------------------------------------------
\section{Average Cluster Geometry}
\label{sec:randomTopology}

This section is devoted to the characterization of the CCDF in (\ref{eq:metric_avg_SIR_distri}). Specifically, tight lower and upper bounds are derived as a function of $K$, $\Nt$ and $\beta$.

%The latter is further utilized in Section \ref{sec:optimalSize} to analyze the ergodic rate and to find the optimal number of cooperating BSs.

\subsection{Upper and Lower Bounds}

The following Lemma provides the probability density function (PDF) of $\delta_1=\frac{\|{\bf d}_1\|}{\|{\bf d}_K\|}$ induced by the underlying PPP.

\begin{lemma}[PPP distance ratio] \label{lemma:distance_ratio} 
Let $\|{\bf d}_1\|$ and $\|{\bf d}_K\|$ denote the distances from the origin to the first and the $K$th BSs. The PDF of $\delta_1=\frac{\|{\bf d}_1\|}{\|{\bf d}_K\|}$ is given by
\begin{align}
f_{\delta_1}(x)&=2(K-1)x (1-x^2)^{K-2} \quad \textrm{for} \quad 0\leq x \leq 1.
\end{align}
\end{lemma}
\begin{proof}
 See Appendix \ref{proof:distance_ratio}.
\end{proof}

Applying the above lemma to (\ref{eq:Cov_upper}) and (\ref{eq:Cov_lower}) we can readily write
\begin{align}
F^{\rm c,U}_{\textrm{SIR}}(K,\Nt,\beta;\gamma) &\!=\!\!\!\! \sum_{\ell=1}^{N_{\textrm{t}}-K+1}\!\!\!\int_{0}^{1}\!\!\!\frac{\binom{N_{\textrm{t}}-K+1}{\ell}(-1)^{\ell+1}}{\left[1+\mathcal{D}(\ell \kappa x^{\beta} \gamma, \beta)\right]^{K}} f_{\delta_1}(x) \d x
\label{eq:Cov_upper_uncondition}  \\
F^{\rm c, L}_{\textrm{SIR}} (K,\Nt,\beta;\gamma)&=\!\!\!\!\! \sum_{\ell=1}^{N_{\textrm{t}}-K+1}\!\!\!\int_{0}^{1}\!\!\!\frac{\binom{N_{\textrm{t}}-K+1}{\ell}(-1)^{\ell+1}}{\left[1+\mathcal{D}(\ell  x^{\beta} \gamma, \beta)\right]^{K}} f_{\delta_1}(x) \d x
 \label{eq:Cov_lower_uncondition}
\end{align}
%where
%\begin{align}
%f_{\delta_1}(u)&=2(K-1)u (1-u^2)^{K-2} \quad \textrm{for} \quad 0\leq u \leq 1.
%\end{align}
for which analytical approximations are derived next.

%Corollary \ref{cor1} yields the upper and lower bounds of the SIR distribution when averaged out over all possible user locations within the cluster. This can be interpreted as the SIR distribution for a typical user in the network when $K$ BSs apply the coordinated beamforming.

%Therefore, using this SIR distribution, it is possible to characterize the ergodic rate of the typical user, which will be used in the next section to optimize the cluster size.

 %
%Due to the integral expression, however, it is difficult to conceive the relevance of the system parameters that change shapes of the CCDF from Corollary \ref{cor1}. 
%
%To obtain some intuition on the averaged performance of the cooperative network, we consider a discrete approximation of $\delta_1\in[0,1]$ with $D$ discretized samples, i.e., $\delta_1\in\{\frac{1}{D},\frac{2}{D},\ldots, \frac{D-1}{D},1\}$. In particular, when $\beta=4$ and by virtue of this approximation, the averaged CCDF of the SIR is given by
%\begin{align}
%{\tilde F}_{\rm SIR}^{c}(K,\Nt,4;\gamma)\simeq \! \sum_{\ell=1}^{M}\!\!\binom{\!M\!}{\ell\!}\!(\!-1)^{\ell+1}\!\!\sum_{d=0}^{D\!-\!1}\!\!\frac{\left[1-\left(\frac{d}{D}\right)^2\right]^{K-1}-\left[1-\left(\frac{d+1}{D}\right)^2\right]^{K-1}}{\left[ 1+\left(\frac{d}{D}+\frac{1}{2D}\right)^2 \sqrt{\ell \kappa \gamma}\arccot\left\{ (\frac{d}{D}+\frac{1}{2D})^{-2} ({\ell \kappa \gamma})^{-1/2} \right\} \right]^K}. \label{eq:approx_avg_CCDF_SIR}
%\end{align}
%This approximation converges to the true CCDF in (\ref{eq:Cov_upper_uncondition}) for $D \rightarrow \infty$.
%
%
%

%\subsection{Validation}

\subsection{Bound Approximations} 
 
As shown in App. \ref{proof:approximation}, the integrals within (\ref{eq:Cov_upper_uncondition}) and (\ref{eq:Cov_lower_uncondition}) satisfy
\begin{align}
\int_{0}^{1}\frac{2(K-1)u (1-u^2)^{K-2} }{\left[1+\mathcal{D}({\tilde \gamma}u^{\beta} , \beta) \right]^{K} } \d u \simeq \frac{1}{1+ \frac{{\tilde \gamma}^{2/\beta}}{\sqrt{K}}\mathcal{A}\left(\frac{\sqrt{K}}{{\tilde \gamma}^{2/\beta}}\right)}
\label{eq:lem4}
\end{align}
where ${\tilde \gamma}=\kappa \ell \gamma$ and $\mathcal{A}\left(y\right)=\int_{y}^{\infty}\frac{1}{1+v^{\frac{\beta}{2}}} \d v$.
From (\ref{eq:lem4}), the bounds to the unconditioned CCDF of the SIR in turn satisfy
\begin{align}
F^{\rm c,U}_{\textrm{SIR}}(K,\Nt,\beta;\gamma) &\simeq \sum_{\ell=1}^{N_{\textrm{t}}-K+1}\!\!\!\frac{\binom{N_{\textrm{t}}-K+1}{\ell}(-1)^{\ell+1}}{1+\frac{{(\kappa \ell \gamma)}^{2/\beta}}{\sqrt{K}}\mathcal{A}\left(\frac{\sqrt{K}}{{(\kappa \ell \gamma)}^{2/\beta}}\right)}
\label{eq:Cov_upper_uncondition_approx}  \\
F^{\rm c, L}_{\textrm{SIR}} (K,\Nt,\beta;\gamma)&\simeq \sum_{\ell=1}^{N_{\textrm{t}}-K+1}\!\!\!\frac{\binom{N_{\textrm{t}}-K+1}{\ell}(-1)^{\ell+1}}{1+\frac{{( \ell \gamma)}^{2/\beta}}{\sqrt{K}}\mathcal{A}\left(\frac{\sqrt{K}}{{( \ell \gamma)}^{2/\beta}}\right)}
\label{eq:Cov_lower_uncondition_approx}
\end{align}
which, for the most relevant case where $\beta=4$ and $N_{\rm t}=K$, coincide yielding
\begin{align}
F^{\rm c}_{\textrm{SIR}}(K,K,4;\gamma) &\simeq\frac{1}{1+\sqrt{  \frac{\gamma}{K}}~\textrm{arccot}\left(\sqrt{\frac{K}{\gamma}}\right)}.
\label{happy2014}
\end{align}
The simple expression in (\ref{happy2014}) clearly shows how coordination with the $K$ nearest BSs improves the CCDF of the SIR with $K$ because $\textrm{arccot}(x)$ is a decreasing function of $x$. Furthermore, this approximation recovers the exact CCDF of the SIR for $K=1$. 

\begin{figure}
\centering
\includegraphics[width=3.8in]{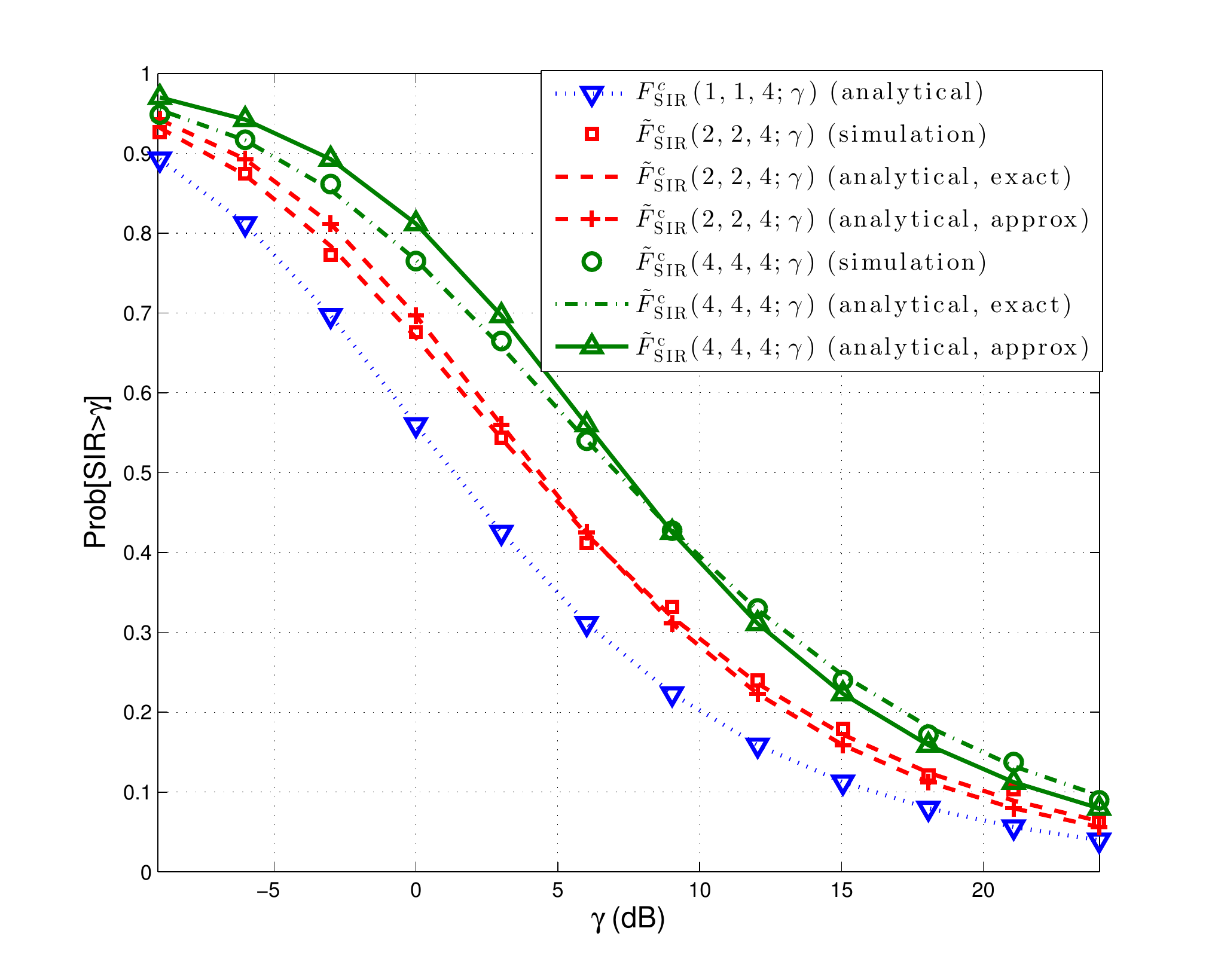}\vspace{-0.3cm}
\caption{CCDF of the SIR for different cluster cardinalities with $\beta=4$ and $\Nt=K$.}
\label{fig:cooperationgain}\vspace{-0.2cm}
\end{figure}

To validate (\ref{eq:Cov_upper_uncondition_approx}) and (\ref{eq:Cov_lower_uncondition_approx}), we compare them with simulation results in Fig. \ref{fig:cooperationgain}. The agreement is excellent for the various values of $K$ considered and the entire range of SIRs of interest. As $K$ increases, the SIR improves because the out-of-cluster interference abates, but that would come at the expense of further overheads. This motivates the optimization of $K$ with the overhead incorporated, a problem that is tackled in Section \ref{sec:optimalSize}.  

\subsection{Low-SIR Analysis}

To wrap up our characterization of the SIR distribution under coordinated beamforming, we specialize it to the low-SIR regime, which is of particular relevance for the purposes of establishing communication outages
in quasi-static communication settings \cite{Lozano2012}.

\begin{proposition} \label{cor3}
The unconditioned CCDF of the SIR for $K=\Nt$ expands as
\begin{align}
F^{\rm c}_{\textrm{SIR}}(K,K,\beta;\gamma) \!=\!  1\!-\!\frac{K(K \!-\! 1)}{\beta -2 }\frac{\Gamma\left(\!\frac{\beta}{2} \!+\! 1\!\right)\!\!\Gamma\left(\!K \!-\! 1\!\right)}{\Gamma\left(\frac{\beta}{2} + K\right)}  \gamma + o(\gamma) \label{eq:cor3}
\end{align}
\end{proposition}

\begin{proof}
See Appendix \ref{proof:outage}. 
\end{proof}

For $\beta=4$, the expansion in (\ref{eq:cor3}) further simplifies into
%\begin{align}
%F^{\rm c}_{\textrm{SIR}}(K,K,\beta;\gamma) & =  \sum_{\ell=1}^{N_{\textrm{t}}-K+1} \binom{N_{\textrm{t}}\!-\!K\!+\!1}{\ell} (-1)^{\ell+1}\left(1-\frac{\kappa \ell \gamma}{K+1}  \right) + o(\gamma)
%\end{align}
%and, for $N_{\rm t}=K$, further as
\begin{align}
F^{\rm c}_{\textrm{SIR}}(K,K,\beta;\gamma) & = 1 - \frac{\gamma}{K+1}  + o(\gamma) 
\end{align}
which evidences that, by allowing coordinated beamforming from the $K$ nearest BSs, the outage probability in a quasi-static setting would decrease linearly with $K$.

%------------------------------------------------------
%------------------------------------------------------
\section{Optimal Cluster Cardinality}
\label{sec:optimalSize}

Having characterized the SIR distributions of a user at a specific relative location and of the typical user, in this section we establish the corresponding optimal (in terms of ergodic spectral efficiency) cluster cardinalities. To compute the effective ergodic spectral efficiency, we incorporate pilot overheads into the formulation so as to account for the cost of acquiring the CSI required for the coordinated beamforming. %and looking at the net spectral efficiency that arises from combining the reduction in interference brought about by larger clusters and the additional overheads that such larger clusters incur.

Although we have ignored background noise in the analysis of SIR distributions, we take it into account for the purpose of pilot transmissions with out-of-cluster interference, as these take place in an orthogonal fashion among users within a cluster. For a given fading coherence $L_{\textrm{b}}$ (in symbols), the pilot overhead is
\begin{align}
\alpha=\frac{L_{\rm p}(K,N_{\textrm{t}},\textrm{SINR})}{L_{\textrm{b}}}
\end{align}
where $L_{\rm p}(K,N_{\textrm{t}},\textrm{SINR})$ denotes the number of symbols reserved for pilots. This number varies with $K$ and $N_{\textrm{t}}$, and also with the pilot-transmission SINR, which can be obtained from a parametric estimation method. We henceforth model the number of pilot symbols as $L_{\rm p}(K,N_{\textrm{t}},\textrm{SINR})=\eta KN_{\textrm{t}}$ where $\eta \geq 1$ is an SINR-dependent parameter that signifies the number of pilots per transmit antenna and fading coherence interval. For instance, utilizing the result in \cite{jindal2010unified}, we can explicitly express $\eta$ as a function of the channel estimation minimum mean square error (MMSE) and the average SINR of the pilot signals, namely
\begin{align} \label{sinr}
\eta&=\max \left\{ 1,  \left\lfloor \frac{1}{ \textrm{SINR}}\left(\frac{1}{\textrm{MMSE}}-1\right)\right\rfloor \right\}.
\end{align}
Hence, the number of pliots $\eta$ approaches $1$ for SINR $\rightarrow \infty$ . 

With the SIR distributions obtained in Secs. \ref{sec:deterministicTopology} and \ref{sec:randomTopology}, we compute the optimal $K$ for different situations.

%average over all cluster geometries,

\subsection{Optimal Cluster Cardinality for a Specific Relative Cluster Geometry}

 Assuming that perfect CSI is gathered from the pilot observations at the receiver and complex Gaussian codebooks are used, the ergodic spectral efficiency (in bits/s/Hz) for a specific relative cluster geometry is
\begin{align}
&C(K,N_{\rm t},\beta,\delta_1,\alpha) \nonumber\\
& = (1-\alpha) \int_{0}^{\infty} \log_2(1+\gamma) \, \d F_{{\rm SIR}| \delta_1}(K,N_{\rm t},\beta,\delta_1;\gamma)  \\
& = (1 - \alpha) \int_{0}^{\infty} \frac{{ F}_{{\rm SIR}| \delta_1}^{\rm c}(K,N_{\rm t},\beta,\delta_1;{\gamma})}{(1+\gamma) \log_e 2} \, \d \gamma
 \label{eq:effective_throuput}
\end{align}
%where ${f}_{{\rm SIR}|\delta_1}(K,N_{\rm t},\beta,\delta_1;\gamma)$ is the PDF of the SIR for a given $\delta_1$ and the second equality follows from integration by parts. 
from which the cluster cardinality $K^\star$ that optimizes the effective spectral efficiency is obtained as the solution of the integer optimization
\begin{align}
K^\star(\delta_1) = \max_{K\in \{1,2,\ldots, N_{\rm t}\}} C(K,N_{\rm t},\beta,\delta_1,\alpha).
\label{eq:opti}
\end{align}
Note that (\ref{eq:opti}) provides the optimal cluster size for a given in-cluster geometry, with such geometry determined by $\delta_1$, which is itself a function of $K$. The dependence of $\delta_1$ with $K$ can be seen from its distribution, shown in Lemma \ref{lemma:distance_ratio}. From this, we can compute the mean value of $\delta_1$ as $\mathbb{E}[\delta_1]=\frac{\sqrt{\pi}\Gamma(K)}{2\Gamma(1/2+K)}\simeq \frac{1}{\sqrt{K}}$, which decreases as $K$ increases. In other words, adding more BSs to the cluster implies that $\|{\bf d}_K\|$ grows and, consequently, that $\delta_1$ diminishes. This dependence is incorporated in the maximization problem by setting $\delta_1=(c/K)^{1/2}$ with $c$ being a constant value which determines the in-cluster geometry independently of $K$.

\begin{figure}
\centering
\includegraphics[width=3.7in]{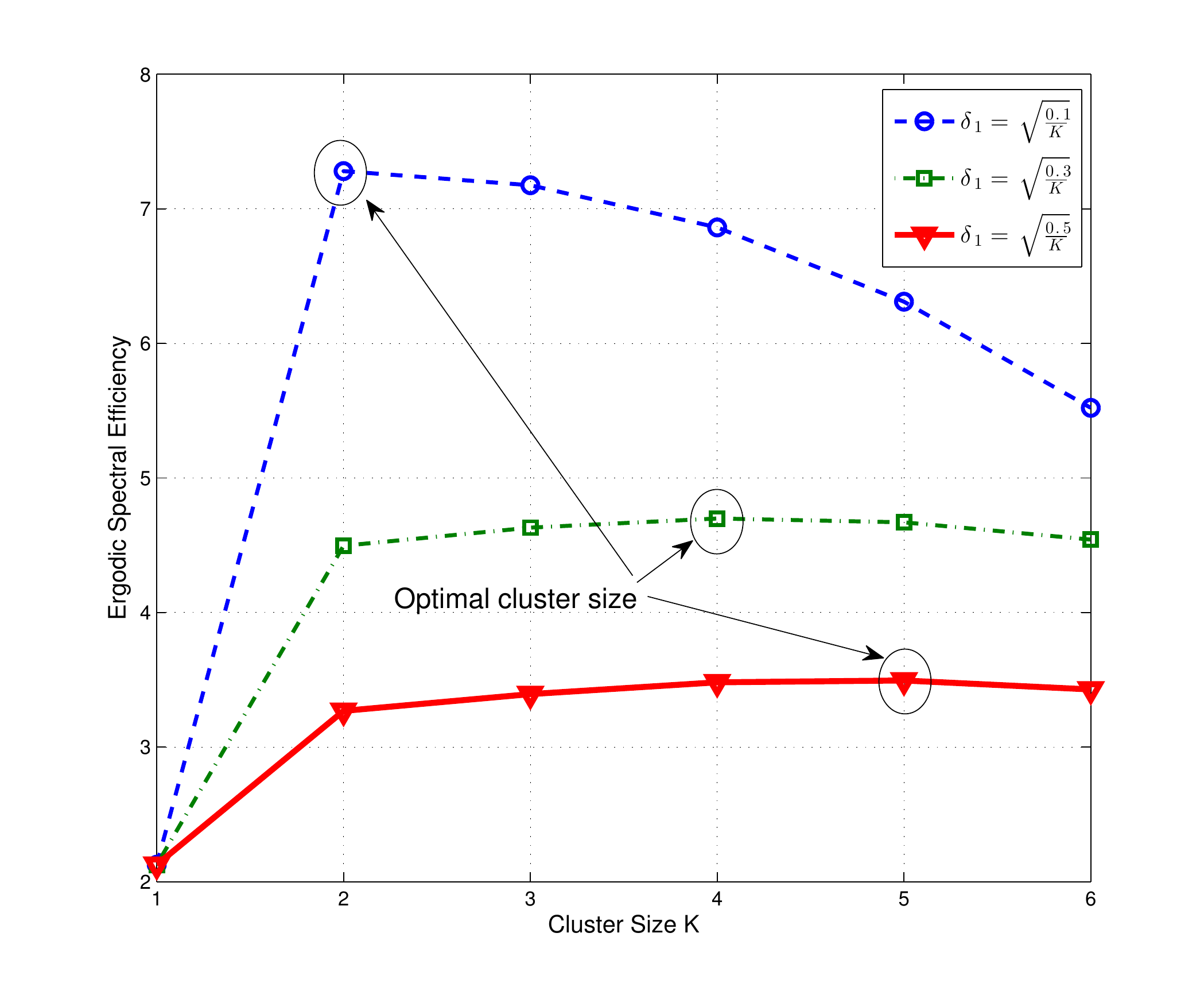}\vspace{-0.3cm}
\caption{Ergodic spectral efficiency as a function of $K$ with a fixed channel coherence parameter $\frac{L_{\rm b}}{\eta}=200$. }
\label{fig:clustersize_fixGeo} \vspace{-0.3cm}
\end{figure}

Having this dependence in mind, Fig. \ref{fig:clustersize_fixGeo} shows the ergodic spectral efficiency as a function of the cluster size when the user's relative location $\delta_1$ is consistently scaled with $K$ as $\delta_1=(c/K)^{1/2}$. With constant values $c\in\{0.1,0.3,0.5\}$, we control the initial size of the protection area, which subsequently scales with $K$. Therefore, as depicted in Fig. \ref{fig:clustersize_fixGeo}, if the user has initially a large protection area (small $c$), the benefit of adding more BSs to the cluster does not overcome the overheads and the optimal cluster size remains small.

%Fig. \ref{fig:clustersize_fixGeo} shows the effective ergodic spectral efficiency for $\delta_1=\left\{ \sqrt{\frac{0.1}{K}},\sqrt{\frac{0.3}{K}},\sqrt{\frac{0.5}{K}}\right\}$ with parameters $L_{\rm b}/\eta=200$, $\beta=4$, and $N_{\rm t} = K$. As shown in this figure, $K^\star$ increases with $\delta_1$ because cluster-edge users benefit the most from a cluster enlargement---even if their spectral efficiencies remain lower.

% a large $\delta_1$ corresponds to users is more affected by out-of-cluster interference, and it is precisely in this case when cooperation becomes more beneficial. This reveals that the user located at the cluster edge obtains higher spectral efficiency gains by increasing $K$, agreeing with intuition.   

%
%Let us consider two different values of $\delta_1\in\{0.3, 0.7\}$ where $\delta_1=0.3$ and $\delta_1=0.7$ correspond respectively to the cases where a user is located at the cluster center and at the cluster edge. For a small target SIR $\gamma=2$ (3 dB) and $K=2$, the coverage probability at the cluster center is $F^{\rm c}_{\textrm{SIR}}(2,2,4,0.3; 2)=0.968$, whereas at the cluster edge we have $F^{\rm c}_{\textrm{SIR}}(2,2,4,0.7; 2)=0.3907$. From this example, we observe that the coverage probability under multi-cell coordination is very sensitive to the in-cluster user location.

\subsection{Optimal Cluster Cardinality for an Average Cluster Geometry}

%We compute the net spectral efficiency averaged over $\delta_1$ to obtain the average efficiency over all possible geometries.
The effective spectral efficiency averaged over $\delta_1$ is 
\begin{align}
C(K,N_{\rm t},\beta,\alpha) & %= (1-\alpha) \int_{0}^{\infty} \log_2(1+\gamma) {f}_{\rm SIR}(K,N_{\rm t},\beta;\gamma) \, \d \gamma \\
%&
= (1 - \alpha) \int_{0}^{\infty} \frac{\log_2 e}{(1+\gamma) } { F}_{\rm SIR}^{\rm c}(K,N_{\rm t},\beta;{\gamma}) \, \d \gamma
 \label{eq:effective_throuput_avg}
\end{align}
%where ${f}_{\rm SIR}(K,N_{\rm t},\beta;\gamma)$ is the PDF of the SIR.
from which the optimal cluster cardinality $K^\star$ for an average cluster geometry is obtained as the solution of the integer optimization
\begin{align}
K^\star = \max_{K\in \{1,2,\ldots, N_{\rm t}\}} C(K,N_{\rm t},\beta,\alpha)
\label{eq:opti_avg}
\end{align}
whose solution can be obtained by a numerical line search technique.

%Unfortunately, the obtained ergodic rate in (\ref{eq:effective_throuput}) is given in an integral form which does not admit a tractable expression so as to solve the maximization problem above. However, the integrals can be evaluated numerically, thus avoiding the need for Montecarlo simulations of the entire network. Therefore,

\begin{figure}
\centering
\includegraphics[width=3.6in]{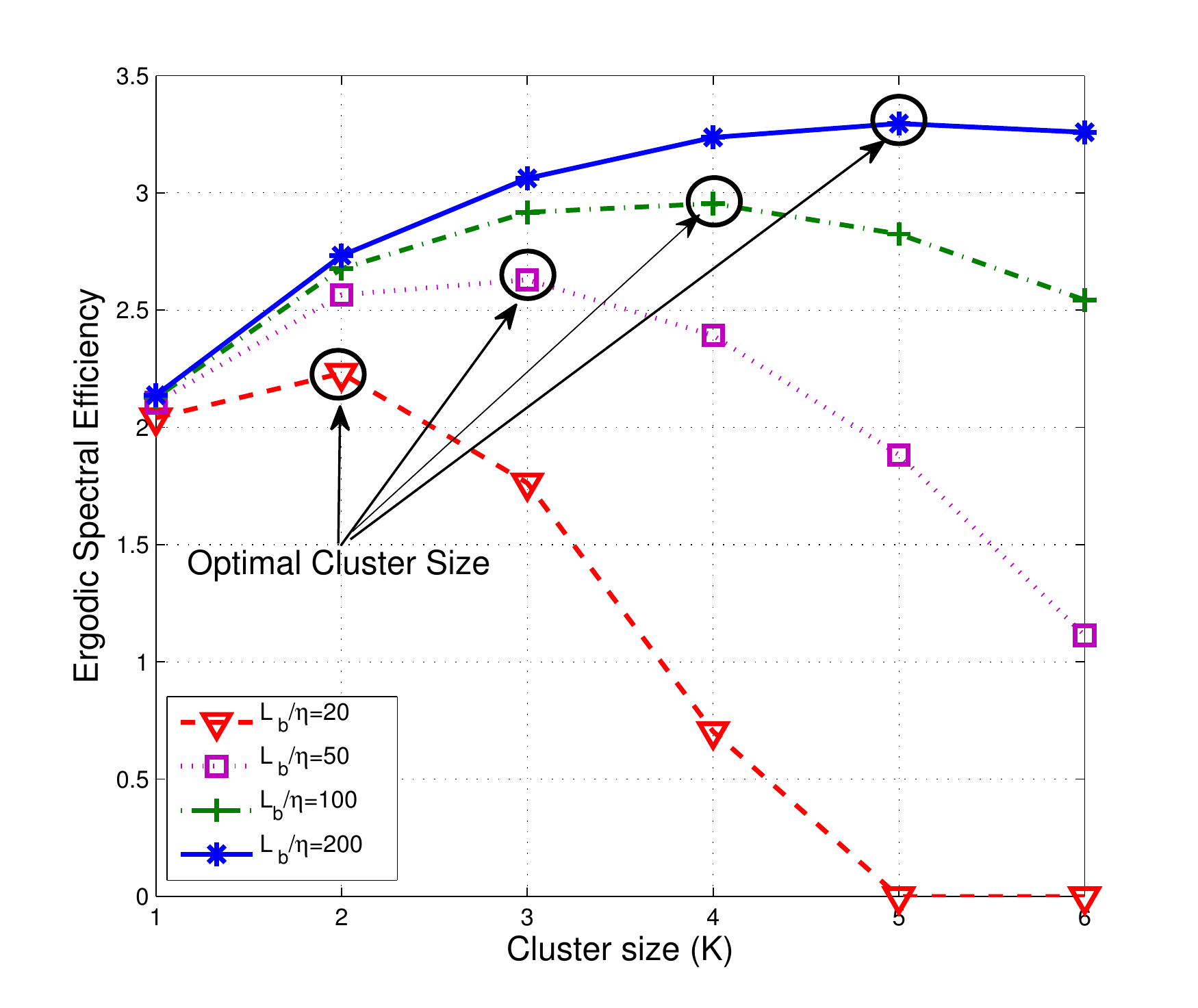}\vspace{-0.3cm}
\caption{Ergodic spectral efficiency as a function of $K$ for different channel coherence parameters, $\frac{L_{\rm b}}{\eta}\in\{20,50,100,200\}$. }\vspace{-0.3cm}
\label{fig:clustersize} 
\end{figure}

Fig. \ref{fig:clustersize} depicts the effective spectral efficiency as a function of $K$ for different ratios $L_{\rm b}/\eta$, with $\beta=4$ and $N_{\rm t} = K$.
%One notable point is that the optimal cluster size is proportional to the channel coherence block length, i.e., inversely proportional to the user mobility.
Notice how the optimum cluster cardinality increases with the channel coherence relative to the pilot cost, ranging from $K^\star=2$ when $L_{\rm b}/\eta=20$ to $K^\star=5$ when $L_{\rm b}/\eta=200$.

\begin{table*}
\caption{Effective Average Spectral Efficiency (bits/s/Hz) of Coordinated Beamforming for $N_{\rm t}=4$ }\vspace{-0.1cm}
\centerline{
    \begin{tabular}{c|c|c|c|c}
	\hline
	Cluster Cardinality $K$ & 1 & 2 (Gain w.r.t. $K=1$) & 3 (Gain w.r.t. $K=1$) & 4 (Gain w.r.t. $K=1$)   \\
	\hline \hline
	No pilot overhead ($\alpha=0$) & 3.968 & 5.018 (26.4$\%$) & 4.249  (7.1$\%$)& 3.517  (-11.4$\%$)    \\ \hline
%	$L_{\rm b}/\eta=20$  & 3.174 (bits/sec/Hz) & 3.474 (9.4$\%$) & 2.508 (-21$\%$)& 1.155 (-64$\%$)     \\
%	\hline
$L_{\rm b}/\eta=200$ $\left(\alpha=\frac{KN_{\rm t}}{200}\right)$ & 3.889 & 4.817  (23.8$\%$) & 3.994 (2.7$\%$)    & 3.236 (-16.8$\%$)     \\
	\hline
	$L_{\rm b}/\eta=20$ $\left(\alpha=\frac{KN_{\rm t}}{20}\right)$ & 3.174 & 3.011 (-5.2$\%$) & 1.699 (-46$\%$)& 0.703 (-78.1$\%$)     \\
	\hline
    \end{tabular}}\label{table2}\vspace{-0.7cm}
\end{table*}

For $N_{\rm t}=K=4$, the effective average spectral efficiencies for different fading coherences are summarized in Table \ref{table2}. Recall that these spectral efficiencies are obtained for the coordinated beamforming method in (\ref{eq:BFsol}), which cancels the $K-1$ nearest interference signals while maximizing the desired signal power with the remaining $N_{\rm t}-K$ degrees of freedom. As shown in the first row of Table \ref{table2}, even without the pilot overheads accounted for, the ergodic spectral efficiency does not necessarily increase with the cluster size $K$ because, once the main interferers have been canceled, beamforming power gain can be more beneficial than the cancellation of additional sources of interference. When the pilot overhead is considered, the optimal cluster cardinality with a short fading coherence ($L_{\rm b}/\eta=20$) is $K^\star = 1$, implying that the cost of the coordination is in this case higher than the return. Alternatively, for a relatively long coherence ($L_{\rm b}/\eta=200$) the optimal cluster cardinality is $K^\star = 2$ with a $23.8\%$ gain in average spectral efficiency relative to the $K=1$ baseline. Altogether, except for very short fading coherence, the optimal ergodic spectral efficiency for $N_{\rm t}=4$ is attained with clusters of size $K=2$ so that the strongest interference signal is canceled and the remaining degrees of freedom (provided by the 2 extra antennas) are put towards maximizing the desired signal strength.

\section{Comparisons}

In this section, we compare the SIR distributions derived under our stochastic geometry model with the corresponding results---obtained through simulation---for a deterministic grid model. Since grid and PPP stochastic geometry models correspond to optimistic and pessimistic scenarios of real BS deployments as argued in \cite{Andrews}, this comparison can be seen to convey upper and lower bounds to the actual benefits of dynamic coordinated beamforming.

%First, let us validate the CCDF derived analytically. Shown in Fig. \ref{fig:cov_comp} is the expression for $\beta=4$ and $K=2$ given in (\ref{eq:Cov_prob_twoBS}) alongside the result of the corresponding Montecarlo simulation. The agreement is complete.

%------------------------------------------------------
%\subsection{Deterministic Grid Model}

In a deterministic square network, the BSs are arranged into a periodic square lattice on a plane. For simulation purposes (cf. Fig. \ref{fig:grid}), we consider 36 BSs located at regular grid points and drop a user uniformly within the highlighted square. Since we are considering dynamic clustering, in each realization the user selects its $K$ nearest BSs for coordination and the rest of the BSs constitute sources of interference. Without loss of generality, we can again index the BSs in increasing distance from the user and express the SIR as
\begin{align}
\textrm{SIR}_{\textrm{grid}}=\frac{H_{1}\|{\bf d}_1\|^{-\beta}}{\sum_{k=K+1}^{36}H_{k}\|{\bf d}_k\|^{-\beta}}
\end{align}
where $H_{k}$ and $\|{\bf d}_k\|$ denote fading coefficient and distance from the $k$th nearest BS to the user, respectively. In particular, $H_1$ has a chi-squared distribution with $2(N_{\rm t}-K+1)$ degrees of freedom. Meanwhile,  the fading coefficients of all the interfering links, $H_{K+1},H_{K+2},\ldots,H_{36}$, are exponential with unit mean. Then, the CCDF of the SIR is 
\begin{align}
{F}^{\rm c}_{\textrm{SIR}_{\rm grid}}(K,\Nt,\beta;\gamma) &= \mathbb{E}\left[\mathbb{P}\left[\textrm{SIR}_{\textrm{grid}}>\gamma \mid \|{\bf d}_1\|, \{\|{\bf d}_k\|,H_k\} \right] \right]
\end{align}
where the expectation is over $\|{\bf d}_1\|$ and $\{\|{\bf d}_k\|,H_k\}$ for $k\in\{K+1,\ldots,36\}$.

\begin{figure}
\centering
\includegraphics[width=3.5in]{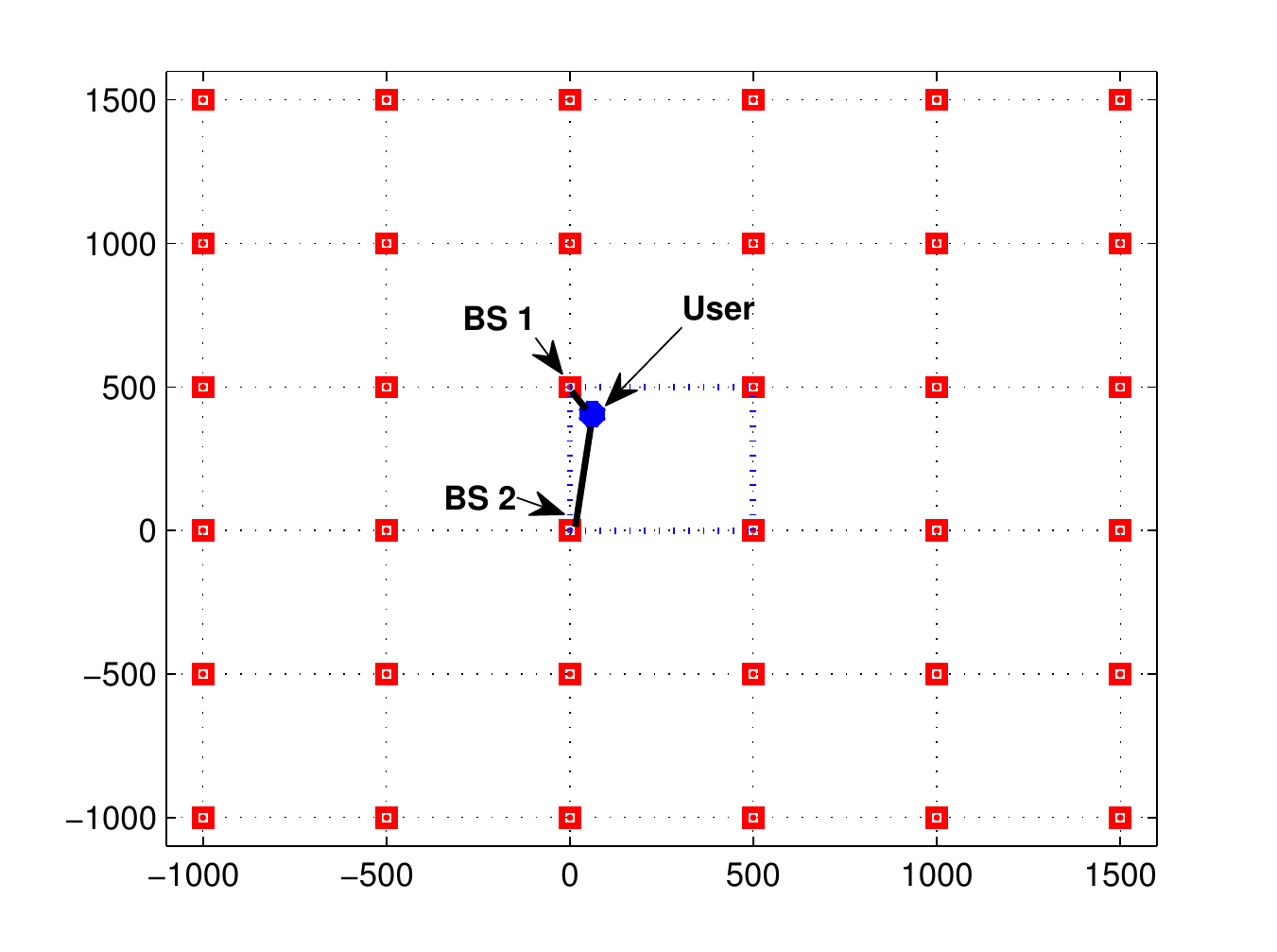}\vspace{-0.3cm}
\caption{A snapshot of the grid model for $K=2$. A user is uniformly located within the square region highlighted in the center, and it selects its two nearest BSs for coordination. In this snapshot, these BSs are the ones located at $(0,0)$ and $(0,500)$.} \label{fig:grid}\vspace{-0.3cm}
\end{figure}

\begin{figure}
\centering
\includegraphics[width=3.5in]{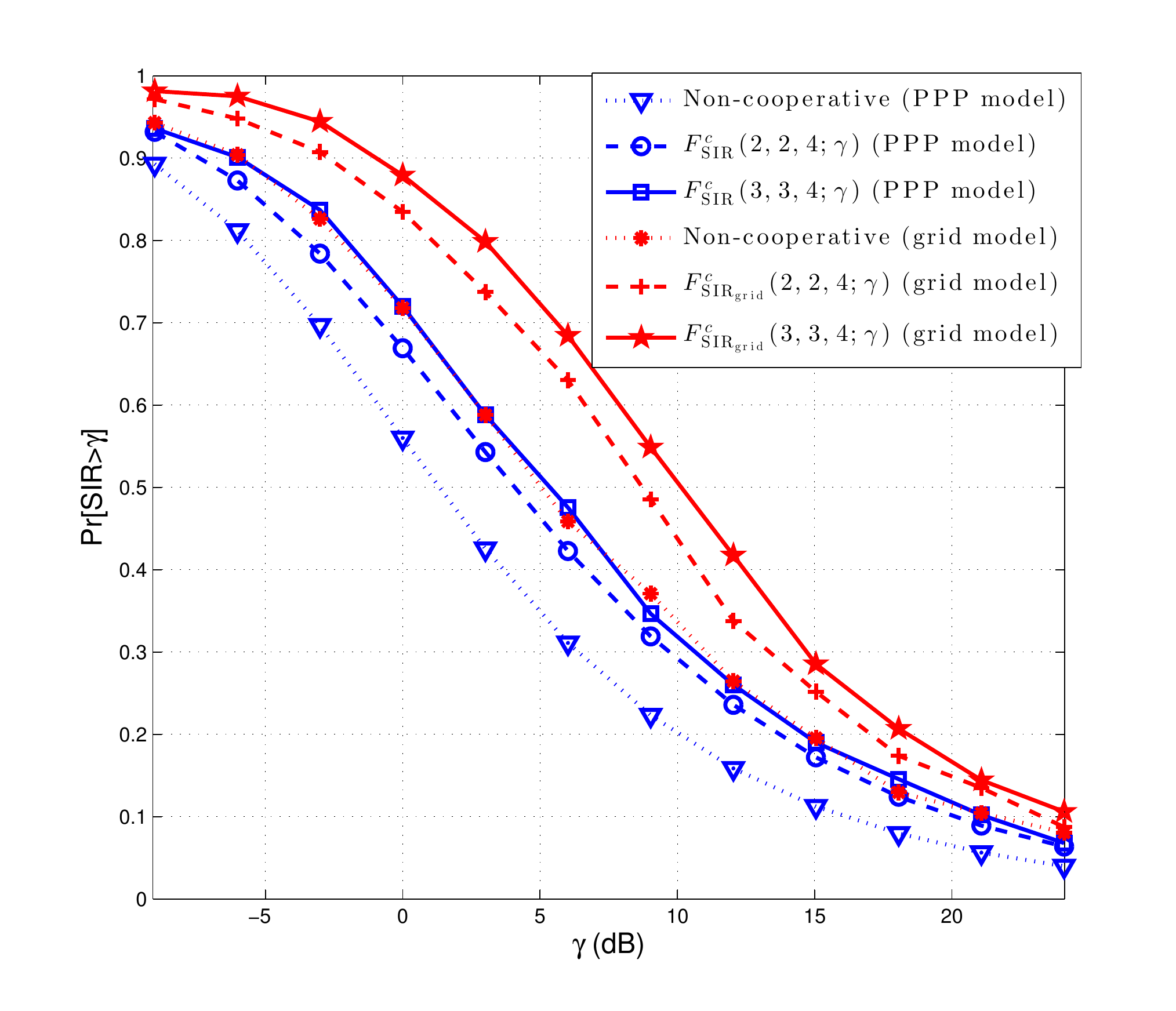}\vspace{-0.3cm}
\caption{CCDF of the SIR with coordinated beamforming for both the stochastic and the deterministic grid models, with $\beta=4$. } \label{fig:cov_comp_CB}\vspace{-0.3cm}
\end{figure}

Fig. \ref{fig:cov_comp_CB} shows the CCDFs for the two different models in which the CCDF of the grid model is obtained by Montecarlo simulation. As one would expect, the CCDF of the SIR in the deterministic grid model is somewhat more favorable than in the stochastic geometry model because, under PPP, the nearest interferer's location can be arbitrarily close to the in-cluster BSs. Nevertheless, the shapes of the CCDFs in the two models are analogous, and they allow us to gauge the potential gains by coordinated beamforming in actual deployments.
\section{Conclusion}

%In this paper, we have characterized the performance of coordinated beamforming with dynamic BS clusters. Capitalizing on tools of stochastic geometry, we have derived the SIR distribution in terms of the number of cooperating BSs, the number of antennas at each BS and the pathloss exponent. From this distribution, we have optimized the number of BS to coordinate as function of the fading coherence. 

In this paper, we have characterized the performance of coordinated beamforming with dynamic BS clusters. Capitalizing on the tools of stochastic geometry, we have derived SIR distributions in terms of the number of BSs per cluster, the number of antennas per BS and the pathloss exponent. From these distributions, we have obtained analytical expressions for the effective ergodic spectral efficiency and optimized the cluster cardinality as function of the fading coherence. Our key finding is that coordinated beamforming is most beneficial to users in the inner part of the coordination clusters as opposed to users near the edges, where the
mitigation of in-cluster interference makes less of a difference because of the strong out-of cluster interference component. Further, we have found that the optimal cluster cardinality for the typical user is small and that it scales with the fading coherence.

\appendices

%------------------------------------------------------
%------------------------------------------------------

\section{Proof of Theorem \ref{Th1}}
\label{proof:Th1}

The proof relies on two Lemmas, reproduced next for the sake of completeness. 

\begin{lemma} \label{lemma:fading} [Fading distribution] The fading distribution of the desired link of a given user, $H_1=|{\bf h}_{1,1} {\bf v}_1|^2$, is chi-squared with $2(N_{\rm t}-K+1)$ degrees of freedom while the fading distributions for the out-of-cluster interference links, $H_{k}=|{\bf h}_{1,k} {\bf v}_{k}|^2$ for $k \in\{K\!+\!1,K\!+\!2,\ldots\}$, are exponential with unit-mean. Furthermore, all fading terms $H_k$ for $k\in\{1,K\!+\!1,K\!+\!2\ldots \}$ are mutually independent. \end{lemma}
\begin{proof} See \cite[Appendix A]{Jindal}.
\end{proof}

\begin{lemma} \label{lemma:distance}[PPP distance distribution] Given a PPP in a plane with intensity $\lambda$, the distribution of the distance $\|{\bf d}_K\|$ between a typical user and its $K$th serving BS is the generalized Gamma distribution
\begin{align}
f_{\|{\bf d}_K\|}(r)= \frac{2(\lambda\pi r^2)^{K}}{r \Gamma(K)} e^{-\lambda \pi r^2} %\label{lemma:distance} ,
\end{align}
where $\Gamma(K)=\int_{0}^{\infty}e^{-x}x^{K-1}\d x$ is the Gamma function.\end{lemma}

\begin{proof}
See \cite{Haenggi}.
\end{proof}

\vspace{3mm}

Let us denote by
\begin{align}
I_{r}&=\sum_{{\bf d}_k\in \Phi/\mathcal{B}(0,r)} H_k\|{\bf d}_k\|^{-\beta}
\end{align}
the aggregate interference power from all the out-of-cluster BSs, conditioned
on the $K$th in-cluster BS location satisfying $\|{\bf d}_K\|=r$ where $\mathcal{B}(0,r)$ denotes a circle centered at the origin with radius $r$. Then, the conditional CCDF of the SIR is given by
\begin{align}
&F^{\rm c}_{\textrm{SIR}|\delta_1} (K,\Nt,\beta,\delta_1;\gamma) \nonumber\\  
&=\mathbb{E}\left[\mathbb{P}\left[\frac{\left( H_1{\delta_1}^{-\beta}\right){r}^{-\beta}}{I_r}\geq \gamma \mid \delta_1 , \|{\bf d}_K\|=r \right] \mid \delta_1 \right]  \\
&= \mathbb{E}\left[\mathbb{P}\left[  H_1 \geq  \delta_1^{\beta}{r}^{\beta}\gamma I_r \mid \delta_1 , \|{\bf d}_K\|=r \right] \mid \delta_1 \right] \\
&= \mathbb{E}\left[\mathbb{P}\left[  H_1 \geq  {r}^{\beta}{\tilde I}_r \mid \delta_1 , \|{\bf d}_K\|=r \right] \mid \delta_1  \right]
\end{align}
where ${\tilde I}_r=\delta_1^{\beta}\gamma I_r$ and the expectation is over the distribution of $r$. From Lemma \ref{lemma:fading}, since $H_1=|{\bf h}_{1,1}{\bf v}_1|^2$ is chi-squared with $N_{\rm t}-K+1$ degrees of freedom, we obtain
\begin{align}
&F_{\rm SIR}^{\rm c}(K,\Nt,\beta,\delta_1;\gamma)\nonumber\\  &=\mathbb{E}\left[ \mathbb{E}\left[\sum_{m=0}^{N_{\rm t}-K} \frac{r^{\beta m}}{m !} {\tilde I}_r^{m}e^{-r^{\beta}{\tilde I}_r}\mid \delta_1 , \|{\bf d}_K\|=r \right] \mid \delta_1\right]
\end{align}
where the inner expectation is over the distribution of ${\tilde I}_r$. From the derivative property of the Laplace transform, which is $\mathbb{E}\left[X^{m} e^{-sX}\right]=(-1)^{m}\frac{\d^m\mathcal{L}_X(s)}{\d s^m}$, we finally obtain
\begin{align}
F_{\rm SIR}^{\rm c}(K,\Nt,\beta,\delta_1;\gamma)\!=\! \mathbb{E}\left[\!\sum_{m=0}^{N_{\rm t}-K}\!\! \frac{r^{\beta m}}{m !}(-1)^{m}\!\! \left.{\frac{\d^m\mathcal{L}_{{\tilde I}_r}(s)}{\d s^m}} \right|_{s=r^\beta} \right]
\end{align}
which completes the proof.

%------------------------------------------------------
%------------------------------------------------------

\section{Proof of Theorem \ref{Th2}}
\label{proof:Th2}

To prove this result, the following lemma is needed. %This Lemma provides us tight bounds for the integral of incomplete Gamma function.

\begin{lemma}
 \label{lemma:Alzer} [Alzer's Inequalty  \cite{Alzer,Huang}] If $H_1$ is chi-squared with $2M$ degrees of freedom, then the CDF $F_{H_1}(\gamma)=\mathbb{P}[H_1< \gamma]$ is upper and lower bounded by
\begin{align}
\left(1-e^{-\kappa\gamma}\right)^{M}\leq F_{H_1}(\gamma) \leq \left(1-e^{-\gamma}\right)^{M}
\end{align}
where $F_{H_1}(\gamma)=\int_{0}^{\gamma}\frac{e^{-x}x^{M-1}}{(M-1)!}\d x$ and $\kappa=(M!)^{-\frac{1}{M}}$. Strict equalities hold when $M=1$, i.e., when $H_1$ is an exponential random variable with mean one.
\end{lemma}

\vspace{3mm}

Now we are ready to prove Theorem \ref{Th2}. We focus on proving the upper bound therein because the lower bound is directly obtained from the former by setting $\kappa=1$.

Conditioned on the $K$th BS being located at distance $r$ from the user, the conditional CCDF of the SIR can be written as
\begin{align}
&\mathbb{P}\left[\frac{\left( H_1{\delta_1}^{-\beta}\right){r}^{-\beta}}{I_r}\geq \gamma \mid \delta_1 , \|{\bf d}_K\|=r \right] 
\nonumber \\&= \mathbb{P}\left[  H_1 \geq  \delta_1^{\beta}{r}^{\beta}\gamma I_r \mid \delta_1 , \|{\bf d}_K\|=r \right].
\label{eq:cov_prob_def}
\end{align}
Applying Lemma \ref{lemma:Alzer} and the binomial expansion,
 \begin{align}
\mathbb{P}[H_1 > x] &\leq 1-\left(1-e^{-\kappa x}\right)^{N_{\rm t}-K+1} \nonumber \\
&=  \sum_{\ell=1}^{N_{\rm t}-K+1} \binom{N_{\rm t}-K+1}{\ell}(-1)^{\ell+1}e^{-\kappa x \ell}
\label{eq:lower_incomplete_gamma}
\end{align}
from which the conditional CCDF of the SIR in (\ref{eq:cov_prob_def}) is upper bounded as
\begin{align}
&\mathbb{P}\left[ H_1 \geq  \delta_1^{\beta}{r}^{\beta}\gamma I_r \mid \delta_1 , \|{\bf d}_K\|=r \right] \nonumber \\
&\leq \sum_{\ell=1}^{M} \binom{M}{\ell}(-1)^{\ell\!+\!1}\! \, \mathbb{E} \left[ \!e^{- \kappa \ell\delta_1^{\beta} {r}^{\beta}\gamma I_r}\! \mid \!\delta_1 , \|{\bf d}_K\|\!=\!r \right]
\end{align}
where the expectation is over the distribution of $I_r$ and $M=N_{\rm t}-K+1$.

Unconditioning with respect to the location of the $K$th BS,
\begin{align}
&F^{\rm c}_{\textrm{SIR}} (K,\Nt,\beta,\delta_1;\gamma) \nonumber\\&=\mathbb{E} \left[\mathbb{P}[\textrm{SIR}(K,\Nt,\beta,\delta_1)>\gamma \mid \delta_1 , \|{\bf d}_K\|=r] \mid \delta_1 \right] \\
& \leq \sum_{\ell=1}^{M} \binom{M}{\ell}(-1)^{\ell+1} \, \mathbb{E} \left[\mathbb{E} \left[e^{-\kappa \ell \delta_1^{\beta} {r}^{\beta}\gamma I_r} \mid \delta_1, \|{\bf d}_K\|=r \right] \mid \delta_1 \right]
\label{eq:Cov_prob_v1}
\end{align}
with inner and outer expectations over the distributions of $I_r$ and $r$, respectively. To evaluate these expectations, we first compute the conditional Laplace transform of $I_r$.
% $\mathcal{L}_{I_{r}}(s)=\mathbb{E} [e^{-sI_{r}} \mid \|x_{K}\|=r]$.
Conditioned on $\|{\bf d}_K\|=r$, such Laplace transform is
\begin{align}
\mathcal{L}_{I_{r}}(s)&=\mathbb{E}\left[e^{-sI_{r}} \mid \|{\bf d}_K\|=r \right] \\
&=\mathbb{E}\left[ e^{-s\sum_{{\bf d}_k\in \Phi/\mathcal{B}(0,r)}H_k\|{\bf d}_k\|^{-\beta}} \mid \|{\bf d}_K\|=r \right]  \\
&= \mathbb{E}\left[ \prod_{{\bf d}_k\in \Phi/\mathcal{B}(0,r)}e^{-sH_k\|{\bf d}_k\|^{-\beta}}\mid \|{\bf d}_K\|=r \right] \label{eqA}  \\
&=  \mathbb{E}\left[ \prod_{{\bf d}_k\in \Phi/\mathcal{B}(0,r)} \frac{1}{1+s\|{\bf d}_k\|^{-\beta}}\mid \|{\bf d}_K\|=r \right] \label{eqB}  \\
&= \exp\left({ -  2\pi\lambda  \int_{r}^{\infty} \frac{u}{1+s^{-1}u^{\beta}} \d u    }\right) \label{eq:conditional_LF}
\end{align}
where (\ref{eqA}) follows from the independence of ${\bf d}_k$ and $H_k$, (\ref{eqB}) holds because $H_k$ is exponentially distributed and unit mean for $k\in\{K+1,K+2,\ldots\}$, and (\ref{eq:conditional_LF}) follows from the probability generating functional of the PPP. Evaluating this conditional Laplace transform at $s=\kappa \ell \delta_1^{\beta}\gamma r^{\beta}$,
\begin{align}
\mathcal{L}_{I_{r}}(\kappa \ell \delta_1^{\beta} \gamma r^{\beta})
&=  \exp\left({ -  2\pi\lambda  \int_{r}^{\infty} \frac{u}{1+(\kappa \ell \delta_1^{\beta}\gamma)^{-1}\left(\frac{u}{r}\right)^{\beta}} \d u    }\right) \nonumber  \\
&\!\!\!\!\!\!\!\!\!\!\!\!\!\!\!\!\!\!\!\!\!\!\!\!\!\!\!\!\!\!=   \exp\left(\!{ -  \pi\lambda r^2 (\kappa \ell \delta_1^{\beta} \gamma )^{2/\beta} \!\!\int_{ (\kappa \ell\delta_1^{\beta} \gamma)^{\frac{-2}{\beta}}}^{\infty} \frac{1}{1+v^{\beta/2}} \d v    }\right) \label{eqAA} \\
&\!\!\!\!\!\!\!\!\!\!\!\!\!\!\!\!\!\!\!\!\!\!\!\!\!\!\!\!\!\!= \exp\left(-  \pi\lambda r^2  \mathcal{D}(\kappa \ell \delta_1^{\beta}\gamma,\beta) \right) \label{eq:conditional_LF_evaluating}
\end{align}
where (\ref{eqAA}) follows from the variable change
\begin{equation}
v=\left[\left(\frac{1}{\kappa \ell \delta_1^{\beta} \gamma }\right)^{1/\beta}\frac{\mu}{r}\right]^2
\end{equation}
while (\ref{eq:conditional_LF_evaluating}) holds because
\begin{align}
\mathcal{D}(\kappa \ell \delta_1^{\beta}\gamma, \beta)&= (\kappa \ell \delta_1^{\beta} \gamma)^{2/\beta} \int_{ (\kappa \ell \delta_1^{\beta} \gamma)^{-2/\beta}}^{\infty} \frac{1}{1+v^{\beta/2}}  \d v \nonumber \\
&\!\!\!\!\!\!\!\!\!\!=\frac{2\kappa \ell \delta_1^{\beta} \gamma}{\beta-2} \, {}_2F_1 \! \left(1,\frac{\beta-2}{\beta},2-\frac{2}{\beta},-\kappa \ell\delta_1^{\beta}\gamma \right) .
\end{align}
%where ${}_2F_1(\cdot)$ denotes Gauss hyper geometric function.
To uncondition the foregoing Laplace transform, we marginalize it with respect to $r$ using the distribution in Lemma \ref{lemma:distance}. With that, the Laplace transform of the aggregate out-of-cluster interference power emerges as
\begin{align}
&\mathbb{E} \left[\mathcal{L}_{I_r}(\kappa \ell \delta_1^{\beta} \gamma  r^{\beta})\right] \nonumber \\ &=\!\int_{r>0}\!\!\! \exp\left({ -  \pi\lambda  r^2\mathcal{D}(\kappa \ell \delta_1^{\beta}\gamma,\beta)  }\right) \frac{2(\lambda\pi r^2)^{K}}{r \Gamma(K)} e^{-\lambda \pi r^2} \d r  \\
\!& = \!\!\!   \int_{0}^{\infty}\!\!\! \!\!\! e^{-x}\!\!\left(\!\!\frac{x}{\pi \lambda[1\!+\!\mathcal{D}(\kappa \ell\delta_1^{\beta}\gamma, \beta)]}\!\!\right)^{\!\!K\!-\!1}\!\!\!\!\!\!\!\!\!\frac{2(\pi\lambda)^{K}\d x}{\Gamma(\!K\!)[\!1\!+\!\mathcal{D}(\kappa \ell\delta_1^{\beta} \gamma, \beta)]2\pi\lambda} 
 \label{eqAAA} \\
& =  \left[\frac{1}{1+\mathcal{D}(\kappa \ell \delta_1^{\beta} \gamma, \beta)}\right]^{K} \label{eq:final_interference_laplace}
\end{align}
where (\ref{eqAAA}) follows from the variable change
\begin{equation}
x=\pi\lambda [1+\mathcal{D}(\kappa \ell \delta_1^{\beta}\gamma,\beta)] r^2
\end{equation}
whereas (\ref{eq:final_interference_laplace}) follows from the definition of the Gamma function. By plugging (\ref{eq:final_interference_laplace}) into (\ref{eq:Cov_prob_v1}), we finally obtain
\begin{align}
F^{\rm c}_{\textrm{SIR}} (K,\Nt,\beta,\delta_1;\gamma) &\leq \sum_{\ell=1}^{N_{\rm t}-K+1}\frac{\binom{N_{\rm t}-K+1}{\ell}(-1)^{\ell+1}}{\left[1+\mathcal{D}(\ell \kappa \delta_1^{\beta} \gamma, \beta)\right]^{K}}
\label{eq:Cov_prob_final}
\end{align}
and, by setting $\kappa=1$, we further have the lower bound in (\ref{eq:Cov_lower}), which completes the proof.

%------------------------------------------------------
%------------------------------------------------------

\section{Proof of Lemma \ref{lemma:distance_ratio}}
\label{proof:distance_ratio}

We start by computing the joint PDF of $\|{\bf d}_1\|$ and $\|{\bf d}_K\|$.
%Let denote a ball centered at the origin with radius $r$ on $\mathbb{R}^2$ by $\mathcal{B}(0,r)$.
Consider the four nonoverlapping areas $A_1=\mathcal{B}(0,r_1)$, $A_2=\mathcal{B}(0,r_1+\d r_1)/\mathcal{B}(0,r_1)$, $A_3=\mathcal{B}(0,r_K)/\mathcal{B}(0,r_1+\d r_1)$, and $A_4=\mathcal{B}(0,r_K+\d r_K)/\mathcal{B}(0,r_K)$. By definition of the PPP, the joint probability that $\|{\bf d}_1\|$ and $\|{\bf d}_K\|$ belong to the two (thin ring) areas $A_2$ and $A_4$, respectively, is given by the product of the four independent probability events as
\begin{eqnarray}
\mathbb{P}\left[\|{\bf d}_1\| \in A_2,\|{\bf d}_K\| \in A_4\right]\!=\! \left\{
\begin{array}{l l}
  P_1 P_2 P_3 P_4& \textrm{if $r_1 \leq r_K $} \\
 0  & \textrm{otherwise}  \label{eq:joint_Prob}
\end{array} \right.
\end{eqnarray}
where
\begin{align}
P_1 &= \mathbb{P}[\textrm{No points in} ~A_1]=e^{-\lambda\pi r_1^2} \nonumber\\
P_2 &= \mathbb{P}[\textrm{One point in} ~A_2]=\lambda \pi 2 r_1\d r_1 e^{-\lambda\pi  2 r_1 \d r_1} \nonumber\\
P_3 &= \mathbb{P}[\textrm{$K-2$ points in} ~A_3] \nonumber \\&=\frac{(\lambda \pi)^{K-2}}{(K-2)!}\left[r_K^2-(r_1+\d r_1)^2\right]^{K-2}e^{-\lambda\pi  2 \left[r_K^2-(r_1+\d r_1)^2\right]}\nonumber \\
P_4 &= \mathbb{P}[\textrm{One point in} ~A_4]=\lambda \pi 2 r_K\d r_K e^{-\lambda\pi  2 r_K \d r_K}. \nonumber
\end{align}
From the limits of the joint probability in (\ref{eq:joint_Prob}), the joint PDF of $\|{\bf d}_1\|$ and $\|{\bf d}_K\|$ emerges as
\begin{align}
&f_{\|{\bf d}_1\|,\|{\bf d}_K\|}(r_1,r_K)=\lim_{\d r_1, \d r_K \rightarrow 0}\frac{\mathbb{P}\left[\|{\bf d}_1\| \in A_2,\|{\bf d}_K\| \in A_4\right]}{\d r_1 \d r_K} \nonumber\\
&=
\left\{
\begin{array}{l l}
 \frac{4(\lambda \pi)^{K}}{(K-2)!}r_1r_K\left(r_K^2-r_1^2\right)^{K-2}e^{-\lambda\pi r_K^2} & \textrm{if $r_1 \leq r_K $} \\
 0  & \textrm{otherwise.} \nonumber
\end{array} \right. \label{eq:joint_PDF}
\end{align}
Utilizing the joint PDF of $\|{\bf d}_1\|$ and $\|{\bf d}_K\|$, we derive the CDF of $\delta_1=\frac{\|{\bf d}_1\|}{\|{\bf d}_K\|}$ as
\begin{align}
\mathbb{P}[\delta_1 \leq x]&=\mathbb{P}\left[\frac{\|{\bf d}_1\|}{\|{\bf d}_K\|}\leq x\right]  \\
&=\mathbb{P}\left[\|{\bf d}_1\| \leq x \|{\bf d}_K\|\right] \\
&=\int_{0}^{\infty}\int_{0}^{x r_K} f_{\|{\bf d}_1\|,\|{\bf d}_K\|}(r_1,r_K)\d r_1\d r_K \\
&\!\!\!\!\!\!\!\!\!\!\!\!\!\!\!\!\!\!\!\!\!\!\!\!\!\!=\int_{0}^{\infty}\int_{0}^{x r_K}   \frac{4(\lambda \pi)^{K}}{(K-2)!}r_1r_K\left(r_K^2-r_1^2\right)^{K-2}e^{-\lambda\pi r_K^2}\d r_1\d r_K \\
&\!\!\!\!\!\!\!\!\!\!\!\!\!\!\!\!\!\!\!\!\!\!\!\!\!\!= 1-(1-x^2)^{K-1}
\end{align}
where $0\leq x \leq 1$. Therefore, the PDF of $\delta_1$ is given by
\begin{align}
f_{\delta_1}(x)&=\frac{\d \mathbb{P}[\delta_1 \leq x]}{\d x}  \\
&=2(K-1)x (1-x^2)^{K-2}.
\end{align}

\section{Proof of Eq. (\ref{eq:lem4})}
\label{proof:approximation}

Recall that 
\begin{align}
\mathcal{D}({\tilde \gamma}\delta_1^{\beta} , \beta)= {\tilde \gamma}^{2/\beta} \delta_1^2\int_{1/({\tilde \gamma}^{2/\beta} \delta_1^2)}^{\infty}\frac{1}{1+v^{\frac{\beta}{2}}} \d v.\label{eq:def_int}
\end{align}
where $\delta_1$ is distributed as per Lemma \ref{lemma:distance_ratio}.
We approximate the integral above as a constant value that captures the effect of the randomness induced by $\delta_1$, 
\begin{align}
\b E\left[ \int_{1/({\tilde \gamma}^{2/\beta} \delta_1^2)}^{\infty}\frac{1}{1+v^{\frac{\beta}{2}}} \d v \right]
\simeq  \frac{1}{\sqrt{K}} \mathcal{A}\left(\frac{\sqrt{K}}{{\tilde \gamma}^{2/\beta}}\right)
\label{eq:approx1}
\end{align}
where the expectation is over $\delta_1$. From (\ref{eq:approx1}),
\begin{align}
\mathcal{D}({\tilde \gamma}\delta_1^{\beta} , \beta) \simeq \frac{{\tilde \gamma}^{2/\beta} \delta_1^2}{\sqrt{K}} \mathcal{A}\left(\frac{\sqrt{K}}{{\tilde \gamma}^{2/\beta}}\right).\label{eq:approx2}
\end{align}
Plugging (\ref{eq:approx2}) into the left side of (\ref{eq:lem4}) and marginalizing with respect to $\delta_1$,
\begin{align}
&\int_{0}^{1}\frac{2(K-1)x (1-x^2)^{K-2} }{\left[1+\mathcal{D}({\tilde \gamma}x^{\beta} , \beta) \right]^{K} } \d x \nonumber \\ &\simeq \int_{0}^{1}\frac{2(K-1)x (1-x^2)^{K-2} }{\left[1+ \frac{{\tilde \gamma}^{2/\beta}}{\sqrt{K}}\mathcal{A}\left(\frac{\sqrt{K}}{{\tilde \gamma}^{2/\beta}}\right) x^2  \right]^{K} } \d x = \frac{1}{1+\frac{{\tilde \gamma}^{2/\beta}}{\sqrt{K}}\mathcal{A}\left(\frac{\sqrt{K}}{{\tilde \gamma}^{2/\beta}}\right)}.
\end{align}

%------------------------------------------------------
%------------------------------------------------------

\section{Proof of Proposition \ref{cor3}}
\label{proof:outage}

Recall that the upper bound in Thm. \ref{Th2} is exact for $K = \Nt$. The outage probability $P_{\rm out}(K,K,\beta,\delta_1;\gamma) = 1 - F^{\rm c}_{\rm SIR}(K,K,\beta,\delta_1 ; \gamma)$ expands at $\gamma = 0$ as
\begin{align}
P_{\rm out}(K,K,\beta,\delta_1;\gamma)& =  1-  F^{\rm c}_{\rm SIR}(K,K,\beta,\delta_1 ; \gamma)  \\
& = 1- \frac{1}{\left(1+\frac{\delta_1^{\beta} \gamma}{\beta-2}\right)^K} + o(\gamma)^2 \\
& = 1-  \left(1-\frac{K \delta_1^{\beta} \gamma}{\beta -2 }\right) + o(\gamma)^2
\end{align}
where we have invoked the series expansion of the Gauss hypergeometric function at $\gamma=0$,
\begin{align}
\mathcal{D}(\delta_1^{\beta} \gamma,\beta) = \frac{\delta_1^{\beta} \gamma}{\beta-2} +o(\gamma)^2
\end{align}
as well as
\begin{align}
\frac{1}{(1+x)^K} =  1-Kx+ o(x)^2.
\end{align}
By dropping the second order error term and marginalizing with respect to $\delta_1$, we obtain the average outage probability
\begin{align}
&1 - F^{\rm c}_{\SIR}(K,K,\beta ; \gamma) \nonumber \\
& = \frac{K \mathbb{E} \left[\delta_1^{\beta}\right]  \gamma}{\beta -2  + o(\gamma)} \\
&= \frac{K \gamma\int_{0}^1\delta_1^{\beta} f_{\delta_1}(x) \d x}{\beta -2 } + o(\gamma) \\
&= \frac{2K(K-1) \gamma\int_{0}^1 x^{\beta+1}(1-x^2)^{K-2} \d x}{\beta -2 } + o(\gamma) \\
&= \frac{K(K-1)}{\beta -2 }\frac{\Gamma\left(\frac{\beta}{2}+1\right)\Gamma\left(K-1\right)}{\Gamma\left(\frac{\beta}{2}+K\right)} \gamma + o(\gamma) \label{eq:asym_analysis}
\end{align}
which completes the proof.

{}


\begin{thebibliography}{1}

%
%\bibitem{C-RAN}
%China Mobile Research Institute,
%\newblock ``C-RAN: The road toward Green RAN,''
%\newblock {\em White Paper,} Oct.
%2011. [Online:] http://labs.chinamobile.com/report/view59826.
\bibitem{ICC2014}
N. Lee, D. Morales, R. W. Heath Jr., and A. Lozano,
\newblock ``Coordinated beamforming with dynamic clustering: a stochastic geometry approach,''
\newblock {\em IEEE Int. Conf. in Communications (ICC'14),} June 2014.


\bibitem{CoMP:12}
D. Lee, H. Seo, B. Clerckx, E. Hardouin, D. Mazzarese, S. Nagata, and K. Sayana,
\newblock ``Coordinated multipoint transmission and reception in LTE-Advanced: deployment scenarios and operational challenges,''
\newblock {\em IEEE Commun. Mag.,} vol. 50, no. 2, pp. 148-155, Feb. 2012.


%\bibitem{3GPP}
%\newblock ``Coordinated multi-point operation for LTE physical layer aspects (Release 11),''
%\newblock {\em 3GPP TR 36.819,} 2011.

\bibitem{Bruno}
B. Clerckx, H. Lee, Y-J. Hong, and G. Kim,
\newblock ``A practical cooperative multi-cell MIMO-OFDMA network based on rank coordination,''
\newblock {\em IEEE Trans. Wireless Communications,} vol. 12, no. 4 pp. 1481 - 1491, Apr. 2013.


\bibitem{Gesbert}
 D. Gesbert, S. Hanly, H. Huang, S. Shamai, O. Simeone, and W. Yu,
\newblock ``Multi-cell MIMO cooperative networks: A new look at interference,''
\newblock {\em IEEE Journal on Sel. Areas in Communications,} vol. 28, no. 9, pp. 1380-1408, Dec. 2010.

\bibitem{Foschini}
M. K. Karakayali, G. J. Foschini, and R. A. Valenzuela,
\newblock ``Network coordination for spectrally efficient communications in cellular systems,'' \newblock {\em IEEE Trans. Wireless Communications,} vol. 13, no. 4,  pp. 56 - 61, Aug. 2006.

\bibitem{Venkatesan}
S. Venkatesan, H. C. Huang, A. Lozano, R. A. Valenzuela,
\newblock ``A WiMAX-based implementation of network MIMO for indoor wireless systems,''
\newblock {\em EURASIP Journal on Advances in Signal Processing,}  vol. 2009, Feb. 2009.

\bibitem{Cadambe_Jafar}
V. R. Cadambe and S. A. Jafar,\newblock ``Interference alignment and the degrees of freedom for the $K$ user interference channel,''
\newblock {\em IEEE Trans. Inform. Theory,} vol. 54, no. 8, pp. 3425 - 3441, Aug. 2008.

\bibitem{Somekh}
O. Somekh, O. Simeone, Y. Bar-Ness, A. M. Haimovich, and S. Shamai,
\newblock ``Cooperative multi-cell zero-forcing beamforming in cellular downlink channels,''
\newblock {\em IEEE Trans. Inform. Theory,} vol. 55, no. 7, pp. 3206 - 3219, Jul. 2009.




\bibitem{Huh}
H. Huh, A. M. Tulino, and G. Caire,
\newblock ``Network MIMO with linear zero-forcing beamforming: Large system analysis, impact of channel
estimation, and reduced-complexity scheduling,''
\newblock {\em IEEE Trans. Inform. Theory,}  vol. 58, no. 5, pp. 2911 - 2934, May 2012.


\bibitem{Lozano_andrew_heath}
A. Lozano, J. G. Andrews, and R. W. Heath, Jr.,
\newblock ``Fundamental limits of cooperation,''
\newblock {\em IEEE Trans. Inform. Theory,} vol. 59, no. 9, pp. 5213-5226, Sept. 2013.


\bibitem{Huang_TWC}
H. Huang, M. Trivellato, A. Hottinen, M. Shaﬁ, P. Smith, and R. Valenzuela,
\newblock ``Increasing downlink cellular throughput with limited network MIMO coordination,''
\newblock {\em IEEE Trans. Wireless Communications,} vol.
8, no. 6, pp. 2983 - 2989, June 2009.

\bibitem{Zhang_Andrews_Heath}
 J. Zhang, R. Chen, J. G. Andrews, A. Ghosh, and R. W. Heath Jr.,
\newblock ``Networked MIMO with clustered linear precoding,''
\newblock {\em IEEE Trans. Wireless Communications,} vol. 8, pp. 1910 - 1921, Apr. 2009.

\bibitem{Xu_Zhang_Andrews}
J. Xu, J. Zhang, and J. G. Andrews,
\newblock ``On the accuracy of the Wyner model in downlink cellular networks,''
\newblock {\em IEEE Trans. Wireless Communications,} vol. 10, pp. 3098 - 3109, Sep. 2011.

\bibitem{Simeone_book}
O. Simeone, N. Levy, A. Sanderovich, O. Somekh, B. M. Zaidel, H. V.
Poor, and S. Shamai,
\newblock ``Information theoretic considerations for wireless
cellular systems: The impact of cooperation,''
\newblock {\em Foundations and Trends
in Communications and Information Theory,}  vol. 7, 2012.

%
%\bibitem{Simeone}
%O. Simeone, O. Somekh, H. V. Poor, and S. Shamai,
%\newblock ``Downlink multi-cell Processing with Limited-Backhaul Capacity,''
%\newblock {\em IEURASIP J. Adv. Signal Process.,} vol. 2009, pp. 840 - 814, 2009.


%
%\bibitem{Wang}
%P. Wang, H. Wang, L. Ping, and X. Lin,
%\newblock ``On the Capacity of MIMO Cellular Systems with Base Station Cooperation,''
%\newblock {\em IEEE Trans. on Wireless Communications,}  vol. 10, no. 11, pp. 3720–3731, Nov. 2011.
%



\bibitem{Andrews}
J. G. Andrews, F. Baccelli, and R. K. Ganti,
\newblock ``A tractable approach to coverage and rate in cellular networks,''
\newblock {\em IEEE Trans. on Communications,} vol. 59, no. 11, pp. 3122 - 3134, Nov. 2011.

\bibitem{Huang_andrews}
K. Huang and J. G. Andrews, \newblock ``An analytical framework for multi-cell coordination via stochastic geometry and large deviations,''
\newblock {\em IEEE Trans. Inform. Theory,}  vol. 59, no. 4, pp. 2501 - 2516, April 2013.

\bibitem{Baccelli}
A. Giovanidis and F. Baccelli, \newblock ``A stochastic geometry framework for analyzing pairwise-cooperative cellular networks,''
\newblock {\em Submitted to IEEE Trans. Inform. Theory}  [online:http://arxiv.org/abs/1305.6254]

\bibitem{Salam}
S. Akoum and R. W. Heath, Jr, \newblock ``Multi-cell coordination: A stochastic geometry approach,''
\newblock {\em IEEE 13th International Workshop on Signal Processing Advances in Wireless Communications (SPAWC)}, pp. 16-20, June. 2012.


%\bibitem{Boccardi_Huang}
%F. Boccardi and H. Huang, \newblock ``Limited downlink network coordination in cellular networks,''
%\newblock {\em IEEE International Symposium on Personal, Indoor and Mobile Radio Commun. (PIMRC)}, pp. 1-5, Sept. 2007.
 


\bibitem{Gesbert_Dynamic}
A. Papadogiannis, D. Gesbert, and E. Hardouin,
\newblock ``A dynamic clustering approach in wireless networks with multi-cell cooperative processing,''
\newblock {\em in Proc. IEEE Int'l Conf. on Comm.}, pp. 4033 - 4037, May 2008.

\bibitem{Mungara}
R. Mungara, G. George, and A. Lozano,
\newblock ``System-level performance of distributed cooperation,''
\newblock {\em in Proc. Asilomar Conf. on Signals, Systems and Computers}, pp. 1561 - 1565, Nov. 2012.

%\bibitem{Mungara14}
%R. Mungara, G. George, and A. Lozano,
%\newblock ``System-level performance of interference alignment,''
%\newblock {\em IEEE Trans. on Wireless Communications}, to appear.

\bibitem{multi-cell_Scheduling}
S. Kaviani and W. A. Krzymien,
\newblock ``Multi-cell scheduling in network MIMO,''
\newblock {\em In Proc. IEEE Global Telecomm. Conf. (GLOBECOM’10)}, pp. 1 - 5, Dec. 2010.


\bibitem{Lee_Heath_Morales_Lozano_1}
N. Lee, R. W. Heath Jr., D. Morales, and A. Lozano,
\newblock ``Base station cooperation with dynamic clustering
in super-dense cloud-RAN,''
\newblock {\em IEEE GLOBECOM Workshop,} Dec. 2013.


\bibitem{Haenggi}
M. Haenggi,
\newblock ``On distances in uniformly random networks,''
\newblock {\em IEEE Trans. Inform. Theory,}  vol. 51, no. 10, pp. 3584 - 3586, Oct. 2005.



\bibitem{Nadarajah}
S. Nadarajah,
\newblock ``A review of results on sums of random variables,''
\newblock {\em ACTA Appl. Math.,} vol. 103, no. 2, pp. 131 - 140, 2008.

\bibitem{Jindal}
N. Jindal, J. G. Andrews, and S. Weber, \newblock ``Multi-antenna communication in ad hoc networks: achieving MIMO gains with SIMO transmission,''
\newblock {\em IEEE Transactions on Communications,} vol. 59, pp. 529 - 540, February 2011.


\bibitem{Huang}
K. Huang, J. G. Andrews, and R. W. Heath, Jr.,
\newblock ``Space division multiple access with a sum feedback rate constraint,''
\newblock {\em IEEE Transactions on Communications,} vol. 55, no. 7, pp. 3879 - 3891, July 2007.

\bibitem{Alzer}
H. Alzer,
\newblock ``On some inequalities for the incomplete Gamma function,''
\newblock {\em Math. Comput.,} vol. 66, no. 218, pp. 771 - 778, 2005.
%
%\bibitem{Lee}
%N. Lee, R W. Heath Jr., D. Morales, and A. Lozano,
%\newblock ``A tractable model for base station cooperation with dynamic clustering,''
%\newblock {\em To appear in IEEE Globecom Workshop,} Dec. 2013. 

\bibitem{jindal2010unified}
N. Jindal and A. Lozano, \newblock ``A unified treatment of optimum pilot overhead in multipath fading channels,''
\newblock {\em IEEE Trans. Communications}, vol. 58, no. 10, pp. 2939-2948, Oct. 2010.

\bibitem{Lozano2012}
A. Lozano and N. Jindal, \newblock ``Are yesterday's information-theoretic fading models and performance metrics adequate for the analysis of today's wireless systems?,''
\newblock {\em IEEE Communications Magazine}, vol. 50, no. 11, pp. 210-217, Nov. 2012.

\end{thebibliography}
\end{document}